\documentclass[12pt,onecolumn,letterpaper]{IEEEtran}
\usepackage{xcolor}
\usepackage{dsfont}
\usepackage{multirow}
\usepackage{setspace}
\usepackage{algorithmic}
\usepackage{graphicx}
\usepackage{amssymb}
\usepackage{amsmath}
\usepackage{theoremref}
\usepackage{epsfig}
\usepackage{epstopdf}
\usepackage{colortbl}
\usepackage{booktabs}
\usepackage{changepage}
\newtheorem{theorem}{Theorem}
\newtheorem{proposition}{Proposition}
\newtheorem{lemma}{Lemma}
\newtheorem{remark}{Remark}
\definecolor{myblue}{gray}{0.9}

\makeatletter
\newcommand{\thickhline}{%
    \noalign {\ifnum 0=`}\fi \hrule height .5pt
    \futurelet \reserved@a \@xhline
}
\newcolumntype{"}{@{\hskip\tabcolsep\vrule width 1pt\hskip\tabcolsep}}
\makeatother

\newtheorem{example}{Example}

\begin{document}

\title{ On Characterization of Elementary Trapping Sets of Variable-Regular LDPC Codes }

%
\author{Mehdi~Karimi,
      and Amir~H.~Banihashemi\\


\authorblockA{Department of Systems and Computer Engineering, Carleton
University, Ottawa, Ontario, Canada }}
\maketitle

\thispagestyle{empty}

 \linespread{1.6}
 \selectfont

\begin{abstract}
In this paper, we study the graphical structure of elementary trapping sets (ETS) of variable-regular low-density parity-check (LDPC) codes.
ETSs are known to be the main cause of error floor in LDPC coding schemes.
For the set of LDPC codes with a given variable node degree $d_l$ and girth $g$, we identify all
the non-isomorphic structures of an arbitrary class of $(a,b)$ ETSs, where $a$ is the number of variable nodes and
$b$ is the number of odd-degree check nodes in the induced subgraph of the ETS.
Our study leads to a simple characterization of dominant classes of ETSs (those with relatively small values of $a$ and $b$)
based on short cycles in the Tanner graph of the code. For such classes of ETSs, we prove that
any set ${\cal S}$ in the class is a layered superset (LSS) of a short cycle, where the term ``layered'' is used
to indicate that there is a nested sequence of ETSs that starts from the cycle and grows, one variable node at a time, to
generate ${\cal S}$. This characterization corresponds to a simple search algorithm that starts from the short cycles of the graph
and finds all the ETSs with LSS property in a guaranteed fashion.
Specific results on the structure of ETSs are presented for $d_l = 3, 4, 5, 6$, $g = 6, 8$ and $a, b \leq 10$ in this paper.
The results of this paper can be used for the error floor analysis and for the design of LDPC codes with low error floors.
\end{abstract}

\section{introduction}
\IEEEPARstart{T}{he} performance of low-density parity-check (LDPC) codes under iterative decoding algorithms in the error floor region is
closely related to the problematic structures of the code's Tanner graph~\cite{Di2002},~\cite{Mackay2002},~\cite{Richardson2003},~\cite{Vontobel2005},~\cite{Laendner2005},~\cite{Milenkovict2007},~\cite{XB-07},~\cite{XB-09},~\cite{Dolecek2007},~\cite{Cole2008},~\cite{Dolecek2009},~\cite{Laendner2009},~\cite{Zhang2009}.
Following the nomenclature of~\cite{Richardson2003}, here, we collectively refer to such structures as {\em trapping sets}.
The most common approach for classifying the trapping sets is by a pair $(a,b)$, where $a$ is the size of the trapping set and $b$ is the number of
odd-degree (unsatisfied) check nodes in the subgraph induced by the set in the Tanner graph of the code.
Among the trapping sets, the so-called \textit{elementary trapping sets (ETS)} are known to be the main
culprits~\cite{Richardson2003}, \cite{Laendner2005}, \cite{Cole2008}, \cite{Milenkovict2007}, \cite{Laendner2009}, \cite{Zhang2009}.
These are trapping sets whose induced subgraph only contains check nodes of degree one or two.

For a given LDPC code, the knowledge of dominant trapping sets, i.e., those that are most harmful, is important.
Such knowledge can be used to estimate the error floor~\cite{Cole2008}, to modify the decoder
to lower the error floor~\cite{Cavus2005},~\cite{Han2008},~\cite{KW-12}, or
to design codes with low error floor~\cite{Ivkovic2008}, \cite{ABAIT2011}.

While the knowledge of dominant trapping sets is most helpful in the design and analysis of LDPC codes,
attaining such knowledge is generally a hard problem~\cite{McGregor2008}.
Much research has been devoted to devising efficient search algorithms for finding small (dominant) trapping sets, see, e.g.,
~\cite{Wang2007}, \cite{Cole2008}, \cite{Rosnes2009}, \cite{XB-09}, \cite{Wang2009}, \cite{Abu2010}, \cite{KW-12}, \cite{KB-12},
and to the (partial) characterization of such sets~\cite{Dolecek2007_1},~\cite{Vasic2009}~\cite{Dolecek2010},~\cite{Laendner2010},~\cite{Huang2011},~\cite{Diao2012}.
Asymptotic analysis of trapping sets has also been carried out in~\cite{Milenkovict2007},~\cite{Abu2008},~\cite{Pusane2009},~\cite{Koller2009},~\cite{Kliewer2009},~\cite{Dolecek2010}.

Laendner \emph{et al.} \cite{Laendner2010} studied the characterization of small $(a,b)$ trapping sets of size up to 8 ($a \le 8$)
and $b/a<1$ in LDPC codes from Steiner Triple Systems (STS). STS LDPC codes are a special category of regular LDPC codes with variable node degree 3.
Huang  \emph{et al.} \cite{Huang2011} showed that for a regular LDPC code with variable node degree $\rho$ and check node degree $\gamma$, and girth~$g \ge 6$,
no $(a, b)$ trapping set with $a \le \rho$ and $b \le \gamma$ can exist. They also studied the trapping sets of Euclidean Geometry (EG) LDPC codes
and provided some bounds on the size and the number of unsatisfied check nodes of the trapping sets of such codes \cite{Diao2012}.
An EG-LDPC code with parameter $q$ is a regular LDPC code of length $q^2$, with variable node degree $q+1$ and check node degree $q$.
A consequence of the bounds derived in \cite{Diao2012} is that for the case where $q=2^s$, there is no trapping set
of size smaller than the minimum distance of the code, i.e., $2^s + 2$, with less than  $2^s + 1$ unsatisfied check nodes.
A subset of trapping sets, called {\em absorbing sets}, for array-based LDPC codes with variable node degrees 2, 3 and 4, were studied in \cite{Dolecek2007_1}, \cite{Dolecek2010}.
Absorbing sets are trapping sets in which each variable node is connected to more satisfied than unsatisfied check nodes
in the induced subgraph of the set~\cite{Dolecek2007_1},~\cite{XB-07}. Array-based LDPC codes are a subclass of (regular) protograph LDPC codes
which are constructed by lifting a fully-connected base graph using cyclic permutations. The analysis in \cite{Dolecek2007_1} and \cite{Dolecek2010}
was focused on minimal absorbing sets, i.e., the ones with the smallest size and with the smallest number of unsatisfied check nodes for a given size.
Vasic \emph{et al.} \cite{Vasic2009} studied the topological structure of trapping sets of size up to $8$ in regular LDPC codes with variable node degree 3,
and proposed a hierarchical search method to find them.

The study of the graphical structure of trapping sets so far has been mainly limited to structured codes, codes with certain variable node degrees, and
to relatively small trapping sets. In this work, for the category of variable-regular LDPC codes with a certain variable node degree and a given girth, we
study the topological structure of $(a,b)$ ETSs for given values of $a$ and $b$, and find all the non-isomorphic structures
of such ETSs. A careful examination of these structures, which are independent of the check node degree distribution of the code,
reveals that for relatively small values of $a$ and $b$, the structures are all
layered supersets (LSS) of small cycles, i.e., they can be characterized by a nested sequence of ETSs which starts from a short cycle and
grows to the ETS one node at a time. The LSS property lends itself to a simple search algorithm that starts from short cycles of the code's
Tanner graph as input and can find all the ETSs with LSS property in a guaranteed fashion. Although the general approach discussed here
can be applied to any category of variable-regular LDPC codes with arbitrary variable node degree $d_l$ and girth $g$ and to any class of ETSs with arbitrary values of $a$ and $b$,
the results presented here are for $d_l = 3, 4, 5, 6$, $g = 6, 8$, and $a, b \leq 10$. One of the main advantages of the results
presented here is that they are applicable to specific codes, rather than just to an ensemble or a category of codes. In particular,
the search algorithm based on LSS property can be used to efficiently find the dominant ETSs of a code in a guaranteed fashion.
This, for example, would imply having a faster and more accurate estimation of error floor for the code under consideration
using techniques such as importance sampling. Moreover, the results presented here can be used in the design of LDPC codes
with low error floor. This can be achieved by avoiding certain dominant ETSs in the Tanner graph of the code. In such a context,
this work can help in identifying the dominant ETSs.

It has been known that dominant trapping sets of LDPC codes have a close relationship with short cycles in the code's Tanner graph~\cite{XB-09},~\cite{XBK-12},~\cite{KB-12}.
This work takes a rigorous step in establishing such a relationship. In general, in comparison with existing results on characterization of trapping sets such
as~\cite{Dolecek2007_1},~\cite{Dolecek2010},~\cite{Vasic2009}, the results presented here are more general in terms of being applicable to both structured and random codes,
and to cover a wider range of variable node degrees and trapping set classes.

The remainder of this paper is organized as follows. Basic definitions and notations are provided in Section II.
In Section III, we present and discuss the LSS property. In Section IV, we develop the algorithm which guarantees to find all
the ETSs with the LSS property starting from the short cycles of the Tanner graph. Sections V, VI, VII, and VIII present
the results for variable-regular LDPC codes with variable node degrees 3, 4, 5 and 6, respectively.
As part of the material presented in these sections, we provide the lengths of short cycles that are required for the proposed algorithm
to find all the $(a,b)$ ETSs with LSS property in a guaranteed fashion, for different values of $a$ and $b$. Section IX is devoted
to discussions and conclusions.

\section{Definitions and Notations}

Let $G=(V=L\cup R\,,E)$ be the bipartite graph, or Tanner graph, corresponding to the
LDPC code $\mathcal{C}$, where $L$ is the set of variable nodes, $R$ is the set of check nodes and $E$ is the set of edges.
The notations $L$ and $R$ refer to ``left" and  ``right," respectively, pointing to the side of the bipartite graph
where variable nodes and check nodes are located, respectively. A \emph{cycle} of length $k$ in a graph $G$ is a
non-empty alternating sequence $v_0e_1v_1 \dots v_{k-1}e_kv_k$ of
nodes and edges in $G$ such that $e_i = \{v_{i-1}, v_i\} \in E$, and $v_i \in L \cup R$ for all $1 \le i \le k$, $v_0 = v_k$, and all the other nodes are distinct.
The length of the shortest cycle in a graph $G$ is denoted by $g$ and
is called the {\em girth} of $G$. The degree of a node $v \in L \cup R$ is
denoted by $d(v)\,$. A bipartite graph is called {\em variable-regular} or \emph{left-regular} with left degree $d_l$ if $d(v)=d_l, \,\forall\, v \in L$.

The graphs $G_1=(V_1\,,E_1)$ and $G_2=(V_2\,,E_2)$ are \emph{isomorphic} if there is a  bijection $f :\, V_1 \rightarrow  V_2$ such that nodes
$v_1$, $v_2 \in V_1$ are joined by an edge if and only if $f(v_1)$ and $f(v_2)$ are joined by an edge.

For a subset $\mathcal{S} \subset L\,$, $\Gamma(\mathcal{S})$ denotes the set of neighbors of $\mathcal{S} $ in $R\,$.
The \textit{induced subgraph} of $\mathcal{S}$, represented by $G(\mathcal{S})$, is the graph containing nodes $\mathcal{S}\cup \Gamma(\mathcal{S})$ with
edges $\{ (u,v) \in E: u \in \mathcal{S} ,\, v\in \Gamma (\mathcal{S})  \}$. The set of check nodes in $\Gamma(\mathcal{S})$ with odd
degree in $G(\mathcal{S})$ is denoted by $\Gamma_{\mathrm{o}}(\mathcal{S})$. Similarly, $\Gamma_{\mathrm{e}}(\mathcal{S})$ represents the set of check nodes in $\Gamma(\mathcal{S})$ with even degree in $G(\mathcal{S})$.
In this paper, we use the terms \textit{satisfied check nodes} and
\textit{unsatisfied check nodes} to refer to the check nodes in $\Gamma_{\mathrm{e}}(\mathcal{S})$ and $\Gamma_{\mathrm{o}}(\mathcal{S})$, respectively.

Given a Tanner graph $G=(L\cup R\,,E)$, a set $\mathcal{S} \subset L$ is called an $(a,b)$ {\em trapping set} if $|\mathcal{S}|=a$ and $|\Gamma_{\mathrm{o}}(\mathcal{S})|=b$.
The integer $a$ is referred to as the {\em size} of the trapping set $\mathcal{S}$.  We also refer to all the trapping sets with
the same parameters $a$ and $b$ as a \emph{class} of trapping sets. An $(a,b)$ trapping set $\mathcal{S}$ is called {\em elementary} if all
the check nodes in $G(\mathcal{S})$ have degree one or two.  A set $\mathcal{S} \subset L$ is called an $(a,b)$ \emph{absorbing set} if $\mathcal{S}$ is an $(a,b)$ trapping set
and if all the nodes in $\mathcal{S}$ are connected to more nodes in $\Gamma_{\mathrm{e}}(\mathcal{S})$ than to nodes in $\Gamma_{\mathrm{o}}(\mathcal{S})$.
Trapping sets with smaller values of $a$ and $b$ are generally more harmful to iterative decoding. Loosely speaking, such trapping sets are called {\em dominant}.

Without loss of generality, we assume that the induced subgraph of a trapping set is connected.
Disconnected trapping sets can be considered as the union of connected ones. Moreover, to the best of our knowledge, almost all the structures
reported as dominant trapping sets (of regular LDPC codes) in the literature have the property that every variable node is
connected to at least two satisfied check nodes in the induced subgraph. We thus focus on trapping sets with this property. In the rest of the paper, we use the notation $\mathcal{T}$ to denote the set of all trapping sets $\mathcal{S} $ in a graph
$G$ whose induced subgraph $G(\mathcal{S})$ is connected and for which every node $v \in \mathcal{S}$ is connected to at least
two nodes in $\Gamma_{\mathrm{e}}(\mathcal{S})$.  In the following, we also assume that
the Tanner graph $G$ has no parallel edges.

\section{Non-isomorphic structures of ETSs}

Elementary trapping sets (ETS) in left-regular Tanner graphs are the main focus of this paper.
To investigate the structure of ETSs of a certain $(a,b)$ class in left-regular Tanner graphs with left-degree $d_l$
and girth $g$, we need to obtain all the non-isomorphic graphical structures of such trapping sets.
To simplify the representation of the subgraph induced by an ETS in a left-regular graph, we often use an alternate graphical representation
called \emph{normal} graphs~\cite{Koetter2004} .
The normal graph of an ETS is obtained from the induced subgraph of the set by removing all the degree-1 check nodes and their edges
from the subgraph, and by replacing each degree-2 check node with an edge.
\begin{example}
Figures~\ref{normal}$(a)$ and \ref{normal}$(b)$ represent the induced subgraph and the normal graph of a $(5,4)$ ETS
in a left-regular graph with $d_l = 4$ and $g = 6$, respectively. In Figure~\ref{normal}$(a)$, and the rest of the paper, variable nodes
are represented by circles, and satisfied and unsatisfied check nodes are shown by empty and full squares, respectively.
\begin{figure}[!h]
\centering
\includegraphics[width=3in]{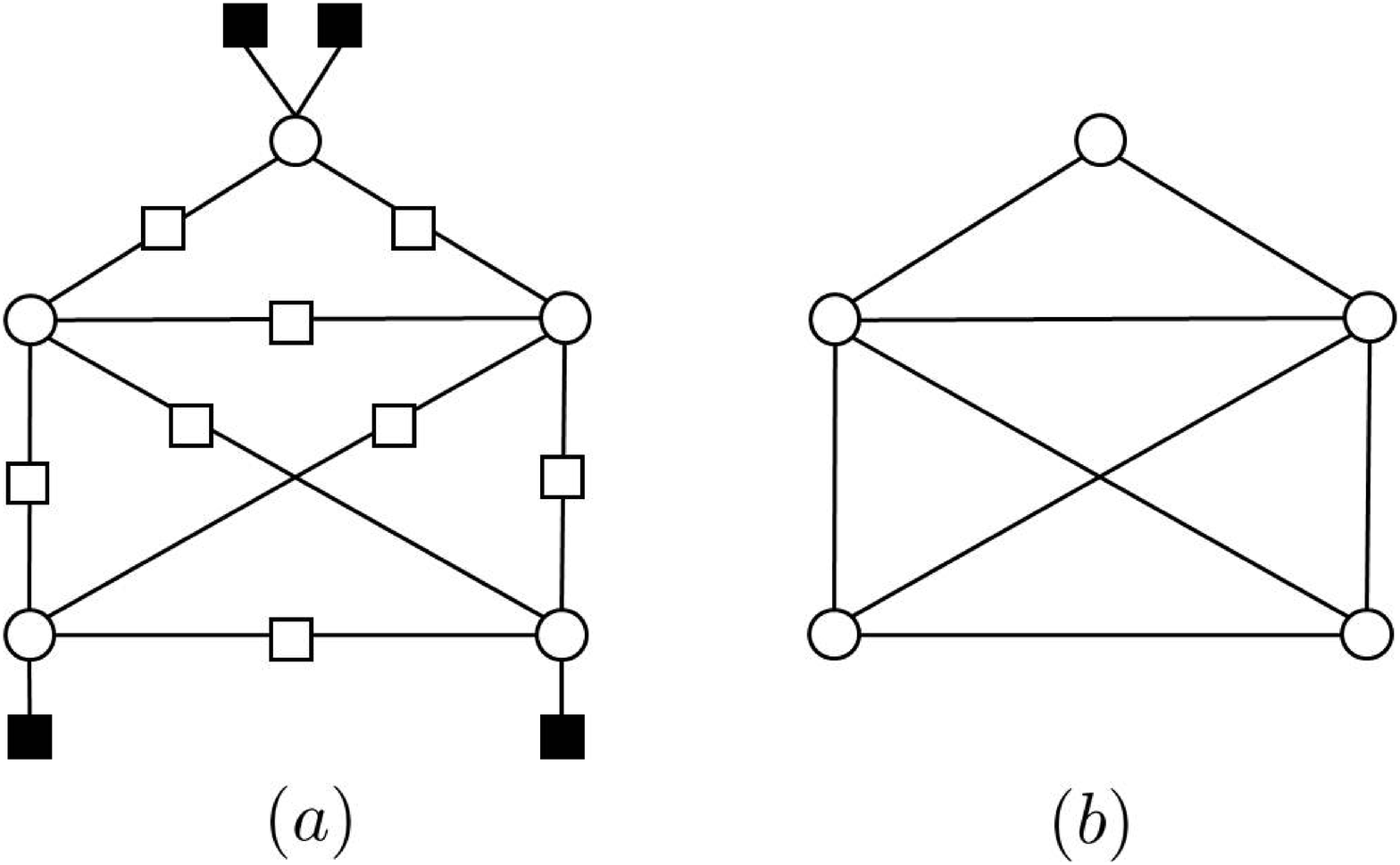}
\caption{A $(5,4)$ elementary trapping set and its normal graph.}
\label{normal}
\end{figure}
\end{example}

It is easy to see that there is a one-to-one correspondence between an ETS of a left-regular graph and its normal graph.
Given $d_l$, one can construct the subgraph of the ETS from a given normal graph by replacing each edge $\{u,v\}$ of the
normal graph with two edges $\{u,c\}$ and $\{c,v\}$, where $c$ is a degree-2 check node
which is also added to the graph, and by connecting $d_l - d(v)$ check nodes of degree one to every variable node $v$ with $d(v) < d_l$.

In the following, we provide an example to demonstrate how all the non-isomorphic structures of a class of ETSs
can be found for rather small values of $d_l$, $a$ and $b$.

\begin{proposition}
Any $(6,2)$ ETS of a left-regular LDPC code with $d_l = 4$ and $g = 6$ has one of the structures presented in Figure~\ref{6_2_G}.
\label{prop1}
\end{proposition}

\begin{figure}[!h]
\centering
\includegraphics[width=5.2in]{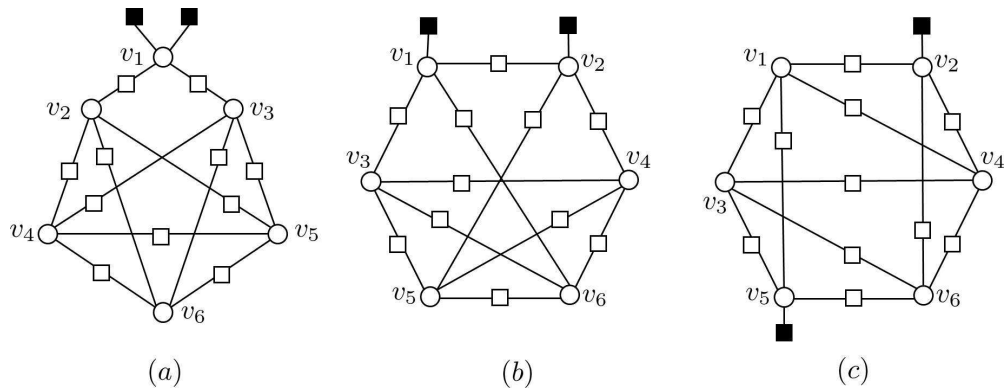}
\caption{All the possible non-isomorphic topologies for $(6,2)$ ETSs in left-regular LDPC codes with $d_l=4$ and $g=6$.}
\label{6_2_G}
\end{figure}
\begin{proof}
We use the normal graph representation, and prove that any $(6,2)$ ETS in a left-regular LDPC code with $d_l=4$ and $g=6$
has one of the normal graph representations given in Figure~\ref{6_2_NG}.

\begin{figure}[!h]
\centering
\includegraphics[width=5in]{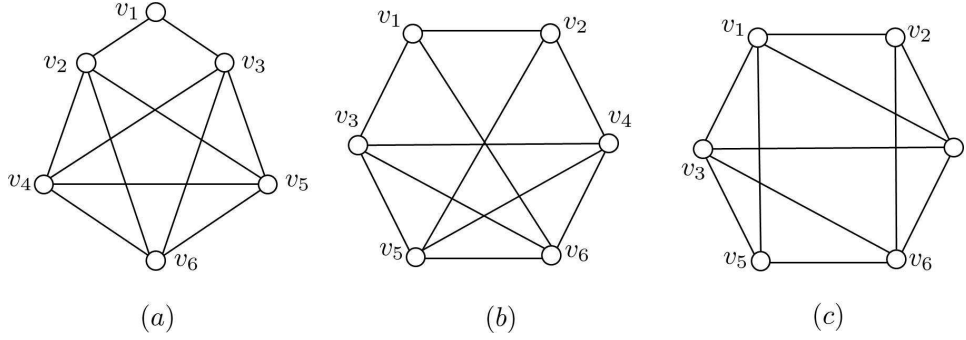}
\caption{All the possible non-isomorphic normal graphs of $(6,2)$ ETSs in left-regular LDPC codes with $d_l=4$ and $g = 6$.}
\label{6_2_NG}
\end{figure}

There are only two possibilities for $(6,2)$ ETSs: (i) Both unsatisfied check nodes are connected to the same variable node, and
(ii) the two unsatisfied check nodes are connected to two different variable nodes. We show that there is only one structure for the
first case and only 2 non-isomorphic structures for the second case.

For Case (i), let $v_1$ be the only variable node connected to two unsatisfied check nodes. This means that in the normal graph representation,
$v_1$ is connected to only two other variable nodes. Suppose that those variable nodes are $v_2$ and $v_3$. Since the degree of all the other variable nodes
$v_4$, $v_5$ and $v_6$ in the normal graph is 4, each of them must have 4 distinct neighbors. The only possibility is for each
of these three nodes to be connected to the other two nodes in addition to nodes $v_2$ and $v_3$.
This also satisfies the degree requirements for $v_2$ and $v_3$ (each of them has degree 4 and is connected to all the nodes $v_4$, $v_5$, $v_6$ and $v_1$).
This topology is shown in Figure~\ref{6_2_NG}$(a)$.

For Case (ii), let $v_1$ and $v_2$ be the two variable nodes, each connected to one unsatisfied check node.
Since all the other 4 nodes have degree 4 and must be connected to 4 distinct nodes, each of them must be connected
to at least one of the nodes $v_1$ and $v_2$. There are two possibilities: a) $v_1$ and $v_2$ are connected together,
and b) $v_1$ and $v_2$ are not connected together. One can show that the only possible topology for Case (a) is the one shown in Figure~\ref{6_2_NG}$(b)$,
and the only possible topology for Case (b) is the one shown in Figure~\ref{6_2_NG}$(c)$.
For Case (a), there always exists a node $v_3$ which is connected to one of the nodes $v_1$ and $v_2$, but not to both.
Otherwise, if all the other 4 nodes are connected to both $v_1$ and $v_2$, both  $v_1$ and $v_2$ will have degree 5 which contradicts
the assumption of the proposition that $d_l = 4$. Without loss of generality, we assume $v_3$ to be connected to $v_1$, and
we grow the normal graph from $v_3$ as the root. Figure~\ref{ab_grow}$(a)$ shows such a graph growth, in which all the nodes except
$v_2$ are located in the first layer and $v_2$ is in the second layer. Node $v_2$ must be connected to two nodes other than $v_1$.
Without loss of generality, we assume those nodes to be $v_4$ and $v_5$ (Figure~\ref{ab_grow}$(b)$). From this point on,
there is no option in connecting the nodes. Node $v_6$ must be connected to 4 distinct nodes, and the only possibility is to
have it connected to all the nodes $v_4$, $v_5$ and $v_1$. This will satisfy the degree of $v_1$, which must be 3.
To satisfy the degrees of $v_4$ and $v_5$ ($d(v_4)=d(v_5)=4$), the only option is to have a connection between them,
which results in the topology shown in Figure~\ref{6_2_NG}$(b)$. The proof that Case (b), i.e., the case where $v_1$ and $v_2$
are not connected, will result in the topology of Figure~\ref{6_2_NG}$(c)$ is similar and is thus omitted.

\begin{figure}[!h]
\centering
\includegraphics[width=3in]{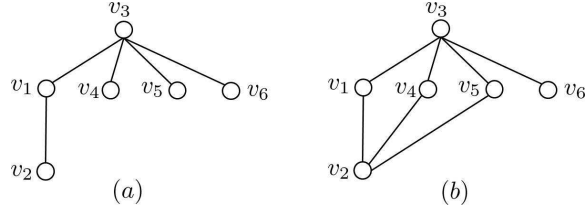}
\caption{Growing the normal graph for Case ii(a) in the proof of Proposition~\ref{prop1}.}
\label{ab_grow}
\end{figure}

\end{proof}

\begin{remark}
The structures in Figures~\ref{6_2_G}$(b)$ and~\ref{6_2_G}$(c)$ are absorbing sets while the one in Figure~\ref{6_2_G}$(a)$ is not.
Moreover, Proposition~\ref{prop1} implies that there are no $(6,2)$ ETSs in left-regular graphs with $d_l = 4$ and $g > 6$.
\end{remark}

Finding the non-isomorphic structures for ETSs with rather large values of $d_l$, $a$ or $b$ can be a formidable task.
We thus resort to software programs to find such structures. One of the well-known software programs related to
graph isomorphism is the \emph{nauty} program~\cite{Online}. This program can be used to efficiently generate all
the non-isomorphic graphs with a given number of nodes (up to 32) and a given number of edges.
The program has many input options to determine the minimum and maximum values of the node degrees,
to select the girth, and to generate only bipartite graphs or all the possible graphs. For the case of bipartite graphs, however, the program does not
have the option of taking the degree distribution of each part of the graph as an input. In the case of finding the non-isomorphic structures (for induced subgraphs) of ETSs,
this limitation results in having a large number of undesired structures at the output of the program. One is thus required to
check all the output structures to find the ones that satisfy the particular degree distributions of the class of ETSs under consideration.
This difficulty can be circumvented by using the normal graph representation of ETSs as explained in the following example.

\begin{example}
Consider the class of $(6,6)$ ETSs in a left-regular graph with $d_l = 4$ and $g = 6$.
Figure~\ref{66_normal}$(a)$ shows one such ETS.

\begin{figure}[!h]
\centering
\includegraphics[width=2.8in]{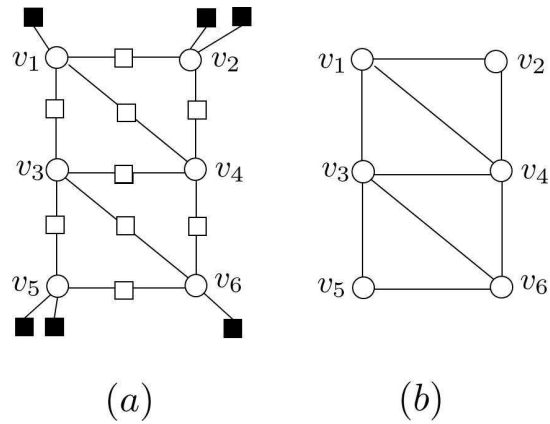}
\caption{A $(6,6)$ elementary trapping set and its normal graph.}
\label{66_normal}
\end{figure}

It is easy to see that the induced subgraph of any ETS in this class has 6 variable nodes of degree 4,
6 check nodes of degree 1, and $((6\times 4)-6)/2=9$ check nodes of degree 2.  This is a total of 21 nodes (check and variable nodes) and 24 edges.
To generate the non-isomorphic structures of $(6,6)$ ETSs using \emph{nauty} package, one way is to set the package parameters to generate
all the connected  bipartite graphs of girth at least 6 with 21 nodes and 24 edges, and with the minimum and maximum degrees of 1 and 4, respectively.
This results in $53,727,932$ graphs, from which an overwhelming majority are not $(6,6)$ ETSs.

Alternatively, we can use the normal graph representation of $(6,6)$ ETSs. Figure~\ref{66_normal}$(b)$ shows the normal graph of
the structure shown in Figure~\ref{66_normal}$(a)$. Similar to Figure~\ref{66_normal}$(b)$, any normal graph of a $(6.6)$ ETS
in a left-regular graph with $d_l = 4$ has 6 nodes and 9 edges. To generate all the non-isomorphic normal structures using the \emph{nauty} program,
one needs to generate all the bipartite and  non-bipartite graphs with 6 nodes and 9 edges, and with the minimum and maximum node degrees of 2 and 4, respectively.
This reduces the number of possible graphs from  $53,727,932$ to only 11. All the 11 graphs correspond to valid normal
structures for the class of ETSs under consideration. These structures are shown in Figures~\ref{66_4_c}, \ref{66_4_a}, \ref{66_4_b}, and \ref{66_4_d}.
Among the 11 structures, only 2 structures are absorbing sets. These are shown in Figure~\ref{66_4_c}.
The other structures are grouped together based on the number of variable nodes with two unsatisfied check nodes.
The structures with 1, 2, and 3 variable nodes connected to two unsatisfied check nodes are shown in Figures~\ref{66_4_a}, \ref{66_4_b}, and \ref{66_4_d}, respectively.

\begin{figure}[!h]
\centering
\includegraphics[width=3.7in]{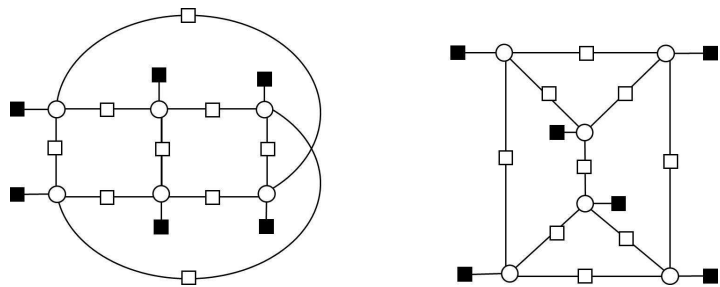}
\caption{Possible topologies for a $(6,6)$ ETS in a left-regular graph with $d_l = 4$ and $g=6$: the only two possible absorbing set topologies.}
\label{66_4_c}
\end{figure}

\begin{figure}[!h]
\centering
\includegraphics[width=4.4in]{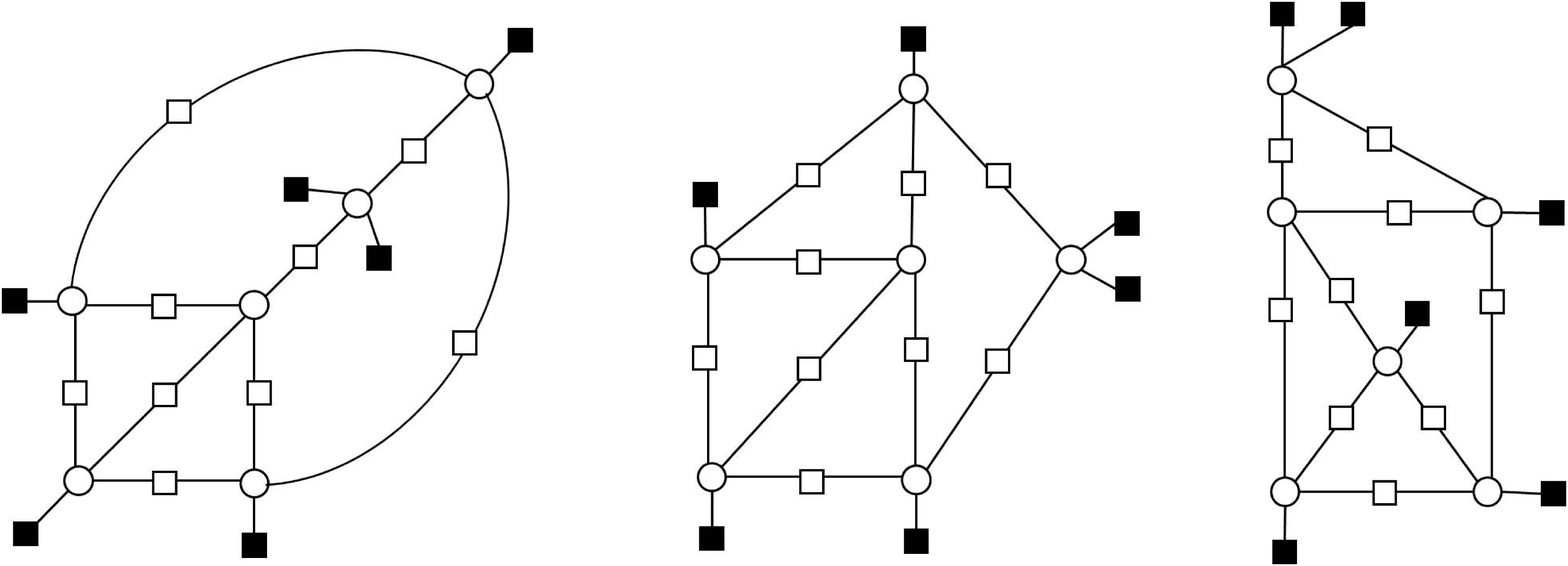}
\caption{Possible topologies for a $(6,6)$ ETS in a left-regular graph with $d_l = 4$ and $g=6$: topologies with only one variable node connected to two unsatisfied check nodes.}
\label{66_4_a}
\end{figure}

\begin{figure}[!h]
\centering
\includegraphics[width=4.4in]{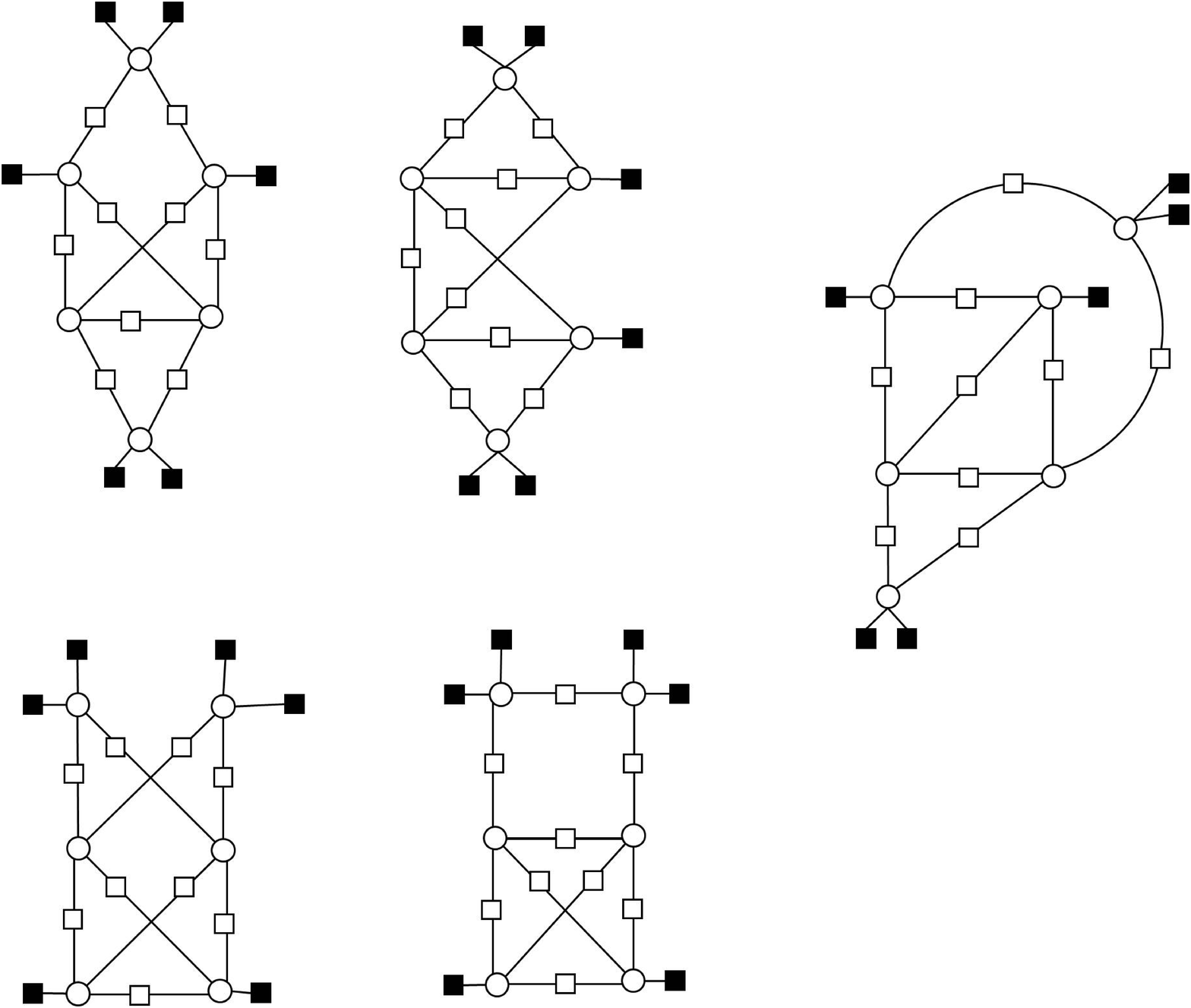}
\caption{Possible topologies for a $(6,6)$ ETS in a left-regular graph with $d_l = 4$ and $g=6$: topologies with only two variable nodes connected to two unsatisfied check nodes. }
\label{66_4_b}
\end{figure}

\begin{figure}[!h]
\centering
\includegraphics[width=1.7in]{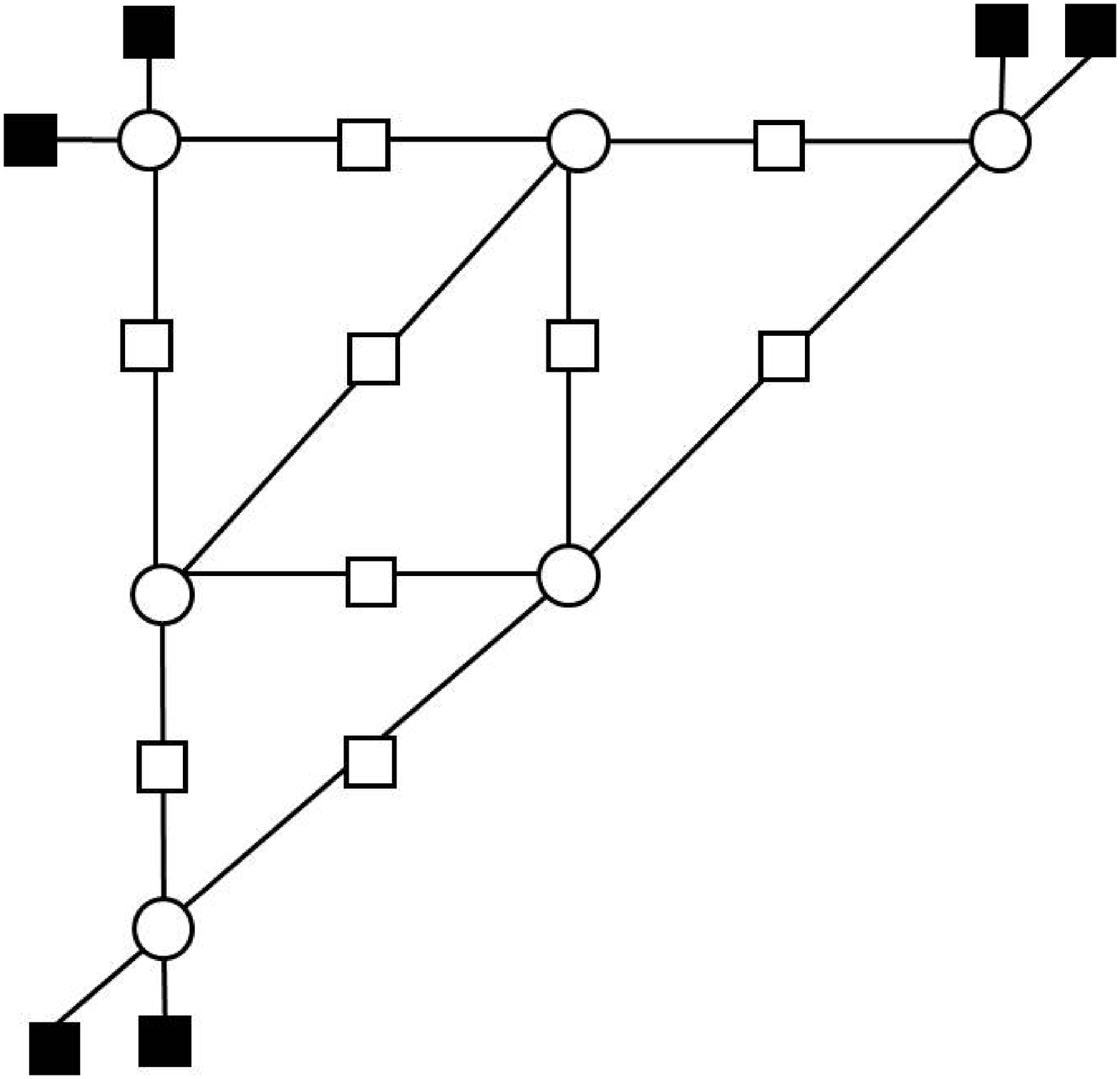}
\caption{Possible topologies for a $(6,6)$ ETS in a left-regular graph with $d_l = 4$ and $g=6$: the only possible topology with three variable nodes
connected to two unsatisfied check nodes.}
\label{66_4_d}
\end{figure}

\end{example}

\section{Layered superset (LSS) property}

Consider an $(a, b)$ ETS $\mathcal{S}$ in $\mathcal{T}$. Let $C \subset \mathcal{S}$ be an ETS in $\mathcal{T}$ of size
$\alpha < a$. We say that $\mathcal{S}$ is a {\em layered superset (LSS)} of $C$ if there exists a nested sequence of
ETSs: $C \stackrel{\Delta}{=} {\cal S}^{(0)} \subset  {\cal S}^{(1)} \subset \cdots \subset {\cal S}^{(a - \alpha)} \stackrel{\Delta}{=} {\cal S}$, such that
${\cal S}^{(i)} \in \mathcal{T}$ has size $\alpha + i$ for $i = 0, \ldots, a - \alpha$. When there is no risk of confusion, we also
refer to ${\cal S}$ as having the LSS property.

\begin{example}
Consider the ETS $\mathcal{S}=\{v_1, v_2, v_3, v_4, v_5, v_6\}$ of a left-regular LDPC code with $d_l = 4$ shown in Figure~\ref{LSS_ex}$(a)$,
and one of its subsets $C_1=\{v_1,v_2,v_3\}$.
\begin{figure}[!h]
\centering
\includegraphics[width=5.3in]{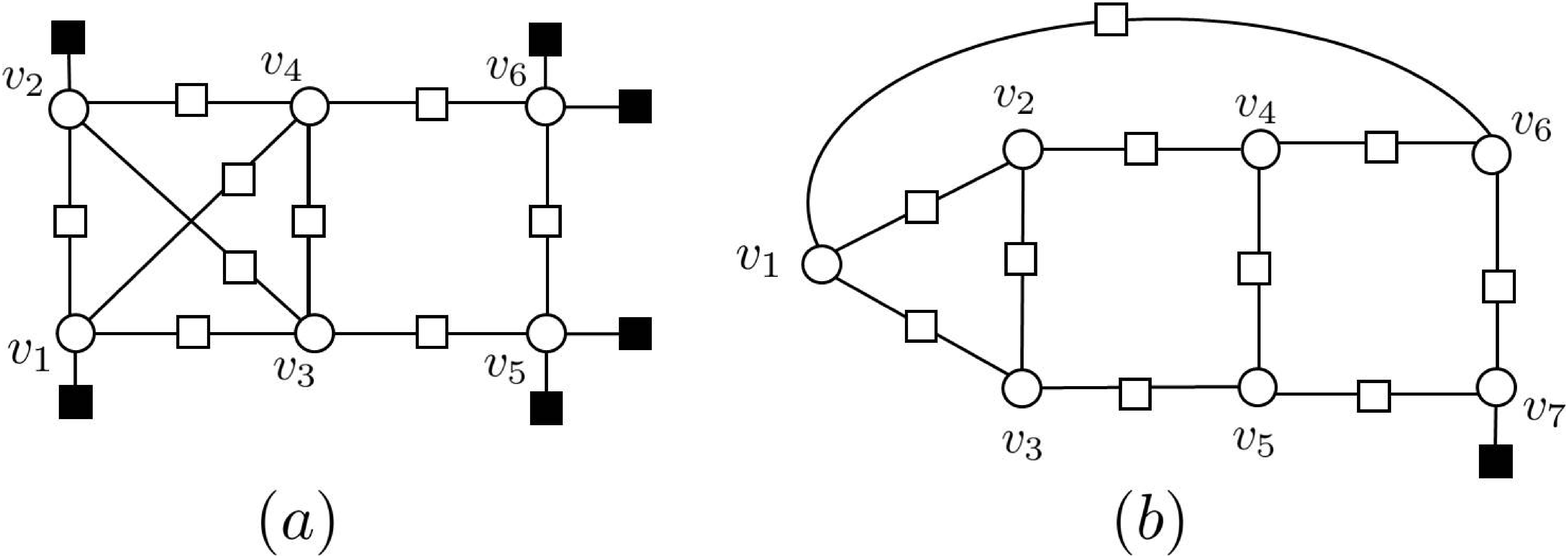}
\caption{(a) A $(6,6)$ ETS in a left-regular graph with $d_l = 4$, (b) A $(7,1)$ ETS in a left-regular graph with $d_l = 3$.}
\label{LSS_ex}
\end{figure}
It is easy to see that $\mathcal{S}$ and $C_1$ are $(6,6)$ and $(3,6)$ ETSs. Careful inspection of Figure~\ref{LSS_ex}$(a)$ reveals that ${\cal S}$
is not an LSS of $C_1$. Set $\mathcal{S}$ however, is an LSS of $C_2 = \{v_3, v_4, v_5, v_6\}$ with the following nested sequence of ETSs:
$C_2 \subset \{v_1, v_3, v_4, v_5, v_6\} \subset {\cal S}$. Figure~\ref{LSS_ex}$(b)$ shows a $(7,1)$ ETS ${\cal S}$ in a left-regular code with $d_l = 3$.
This set has LSS property with respect to $C = \{v_2, v_3, v_4, v_5\}$ with the following nested sequence of ETSs:
$C \subset \{v_1, v_2, v_3, v_4, v_5\} \subset \{v_1, v_2, v_3, v_4, v_5, v_6\} \subset {\cal S}$.
\end{example}

One should note that any cycle in a Tanner graph is an ETS. In the sequel, we are particularly interested in
the LSS property of more complex ETSs with respect to short cycles of the graph. The following proposition is an example.

\begin{proposition}
All the $(6,2)$ ETSs of a left-regular LDPC code with $d_l = 4$ and $g = 6$ are layered supersets of any one of their 6-cycle subsets.
\label{prop2}
\end{proposition}
\begin{proof}
Based on Proposition~\ref{prop1}, there are only 3 non-isomorphic structures for $(6,2)$ ETSs in left-regular LDPC codes with $d_l = 4$ and $g = 6$.
These structures are shown in  Figure~\ref{6_2_G}. It is not difficult to see that all three structures are LSS of any of their 6-cycle subsets.
For example, the structure $\mathcal{S}$ in Figure~\ref{6_2_G}$(b)$ is an LSS of its 6-cycle subset  $C=\{v_1,v_3,v_6\}$
with the following nested sequence of ETSs: $C \subset \{v_1,v_3,v_6,v_5\} \subset \{v_1,v_3,v_6,v_5, v_4\} \subset \mathcal{S}$.
\end{proof}

One advantage of LSS property is that it corresponds to a simple algorithm for finding larger ETSs with LSS property
starting from one of their subsets in the nested sequence of ETSs. The basic step is explained in the following lemma.

\begin{lemma}
Consider an ETS ${\cal S} \in {\cal T}$ of size $a + 1$. Suppose that ${\cal S}$ has an elementary trapping subset ${\cal S}' \in {\cal T}$ of
size $a$. Then, the variable node $v \in {\cal S} \setminus {\cal S}'$ is only connected to unsatisfied check nodes
of ${\cal S}'$ (at least two of them), i.e., there is no connection between $v$ and the satisfied check nodes of ${\cal S}'$.
\label{lemma1}
\end{lemma}
\begin{proof}
The proof is simple and is based on the definitions of an ETS and the set ${\cal T}$.
\end{proof}

The pseudo code of an algorithm corresponding to Lemma~\ref{lemma1} is given in Routine 1.

\linethickness{0.275mm}
\hspace{-.6cm} \line(1,0){500} \\
\textbf{Routine 1: Expansion of an ETS $\mathcal{S}$ of size $a$ to ETSs of size $a+1$ \\
Routine EX = OneExpansion$(\mathcal{S})$}\\
\linethickness{0.125mm}
\line(1,0){500}
\begin{algorithmic}[1]
\STATE  \textbf{Initialization:} \text{EX} $\leftarrow \emptyset$.
\STATE Let $\Gamma_{\mathrm{o}}(\mathcal{S})$ and $\Gamma_{\mathrm{e}}(\mathcal{S})$ be the set of unsatisfied check nodes and satisfied check nodes of $\mathcal{S}$, respectively.
\STATE  Let $\mathcal{O}_2(\mathcal{S})$ be the set of variable nodes which have at least two connections with  the check nodes in  $\Gamma_{\mathrm{o}}(\mathcal{S})$ and have no connection with  the check nodes in  $\Gamma_{\mathrm{e}}(\mathcal{S})$ .
\FOR {each element $v $ in $\mathcal{O}_2(\mathcal{S})$}
\STATE $\mathcal{S}'\leftarrow \mathcal{S} \cup v$.
\STATE $\text{EX} \leftarrow \text{EX} \cup \mathcal{S}'$.
\ENDFOR
\STATE  \textbf{Output:} EX 
\end{algorithmic}
\linethickness{0.275mm}
\line(1,0){500}

\begin{remark}
Based on Lemma~\ref{lemma1}, the set $\mathcal{O}_2(\mathcal{S})$ contains \textit{all} the nodes that can be part of the expansion of the ETS.
\end{remark}
\begin{remark}
Complexity of Routine 1 for expanding an $(a,b)$ ETS to ETSs of size $a+1$ in a $(d_l,d_r)$ regular graph, where $d_l$ and $d_r$ are left and right degrees, respectively,
is of order $O(b d_r)$. This comes from the fact that for finding the larger trapping sets, the algorithm needs to check at most
$b (d_r - 1)$ variable nodes as possible candidates. The memory requirement in this case is of order $O(a b d_r)$.
It should be however noted that, imposing the condition of Line 3 removes a large portion of neighboring variable nodes of $\mathcal{S}$ from the set
of possible candidates for expansion. This highly reduces the complexity of the algorithm.
\end{remark}

Consider the case where an ETS $\mathcal{S}$ of size $a$ is an LSS of an ETS $C$ of size $\alpha$.
Clearly, starting from $C$, the successive application of Routine 1 will result in finding all the ETSs which are layered supersets of
$C$. In particular, ETS $\mathcal{S}$ will be among the outputs after $a - \alpha$ applications of Routine 1.
Algorithm 1 contains the pseudo code of an algorithm that starts from a set of ETSs and finds {\em all} the ETSs
of size up to $k$ that are layered supersets of the initial set of ETSs.

\hspace{-.25in}
\linethickness{0.275mm}
\line(1,0){250} \\
\textbf{Algorithm 1}: Expansion of input ETSs to ETSs of size up to $k$ in $G=(L\cup R\,,E)\,$.\\
($\mathcal{L}_{\mathrm{in}}$ and $\mathcal{L}_{\mathrm{out}}$ are the lists of input and output trapping sets, respectively.)\\
\linethickness{0.125mm}
\line(1,0){250}
\begin{algorithmic}[1]
\STATE  \textbf{Inputs:} $G$ , ${\cal L}_{\mathrm{in}}$, $k$.
\STATE  \textbf{Initialization:} $\mathcal{L}^{i}_{\mathrm{out}}=\{\mathcal{S}\in {\cal L}_{\mathrm{in}}, |\mathcal{S}|=i$ \} for
$i = 2, \ldots, k$.

\FOR { each $i<k$ starting from the smallest one}
\FOR { each element $\mathcal{S} \in \mathcal{L}^{i}_{\mathrm{out}}$ }
\STATE EX= OneExpansion$(\mathcal{S})$
\STATE $\mathcal{L}^{i+1}_{\mathrm{out}} \leftarrow \mathcal{L}^{i+1}_{\mathrm{out}} \cup$ EX.
\ENDFOR
\ENDFOR
\STATE  \textbf{Output:} $\mathcal{L}_{\mathrm{out}}= \mathcal{L}^{2}_{\mathrm{out}} \cup \cdots \cup  \mathcal{L}^{k}_{\mathrm{out}}$.
\end{algorithmic}
\linethickness{0.275mm}
\line(1,0){250}

\section{Characterization of Elementary Trapping Sets of Left-Regular LDPC Codes}

\subsection{Motivating Examples}

It is easy to see that any trapping set in ${\cal T}$, including ETSs, contains at least one cycle.
It can thus be argued that cycles are the most basic structure for ETSs in ${\cal T}$. This motivates the study
of the relationship between cycles, as the most basic ETSs, and the more complex ETSs of ${\cal T}$.
Short cycles have long been known to be problematic for iterative decoding in general~\cite{Mao-Banihashemi-ICC-01}, and for the performance
in the error floor region, in particular~\cite{XB-09},~\cite{XBK-12},~\cite{KB-12}.
Short cycles are also easy to enumerate~\cite{XB-09}. In this section, we establish a simple relationship
between short cycles and the more complex ETSs. In particular, we prove that an overwhelming majority of dominant ETSs
of left-regular LDPC codes are layered supersets of short cycles. One important implication of this result
is that, starting from short cycles of the graph, Algorithm 1, presented in the previous section, can be used to find
all such ETSs in a {\em guaranteed} fashion.

The following examples demonstrate the relationship between short cycles and the more complex ETSs of left-regular graphs.

\begin{example}
Figure~\ref{51_3} shows a possible structure ${\cal S}$ of a $(5,1)$ ETS in a left-regular graph with $d_l = 3$ and $g = 6$.
The $(3,3)$ ETS $C=\{v_1, v_2, v_3\}$ forms a 6-cycle and is a subset of ${\cal S}$.
Set ${\cal S}$ is an LSS of $C$ with the following nested sequence of ETSs: $C \subset {\cal S}^{(1)} = \{v_1, v_2, v_3, v_4\} \subset {\cal S}$.
Starting from $C$, in the first round of expansion by Algorithm~1, variable node $v_4$ will be added and the $(4,2)$ ETS ${\cal S}^{(1)}$ will be found.
In the second round of expansion, with set ${\cal S}^{(1)}$ as the input, variable node $v_5$ will be added and $\mathcal{S}$ will be found.
\begin{figure}[!h]
\centering
\includegraphics[width=1.6in]{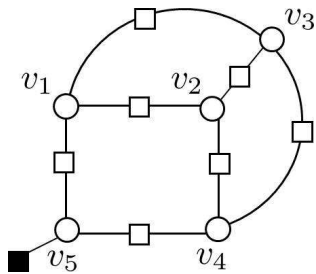}
\caption{A $(5,1)$ trapping set in a left-regular graph with $d_l = 3$.}
\label{51_3}
\end{figure}

Note that although $v_5$ is not part of any 6-cycle, this node will be added at the second round of expansion. In other words,
for a trapping set to be found by the algorithm, it is not necessary that all the nodes in the set participate in short(est) cycles.
\end{example}

Starting from a short cycle, each round of expansion by Algorithm~1, may result in several new ETSs.
These trapping sets all have the same size, but may have different number of unsatisfied check nodes,
and thus belong to different classes of trapping sets. This is demonstrated in the following example.

\begin{example}
Consider the structure in Figure~\ref{6_2_G}$(b)$ and one of its subsets $C=\{v_3,v_4,v_6\}$ which is a $(3,6)$ ETS.
Starting from $C$, as the input to Algorithm~1, there are two variable nodes $v_1$ and $v_5$ as candidates for
the first expansion. The resultant ETSs are $\{v_3,v_4,v_6,v_1\}$ and $\{v_3,v_4,v_6,v_5\}$, which are
$(4,6)$ and $(4,4)$ trapping sets, respectively.
\label{2b}
\end{example}

It is important to note that if an element of a class of ETSs is an LSS of a short cycle, this does not necessarily mean that all
the elements of that class are also LSSs of short cycles of the same length. The reason is that there
may be other non-isomorphic structures in that class which are not LSSs of any of the short cycles under consideration.

\begin{example}
Figure~\ref{66_4} shows two possible structures for a $(6,6)$ ETS in a left-regular graph with $d_l = 4$ and $g = 6$.
The structure of Figure~\ref{66_4}$(a)$ is an LSS of any of its 6-cycles and thus can be obtained by the recursive expansion of any of them through Algorithm~1.
The structure of Figure~\ref{66_4}$(b)$, on the other hand, is not an LSS of any of its 6-cycles, and is thus out of the reach of Algorithm~1,
if the algorithm starts from any 6-cycle.

\begin{figure}[!h]
\centering
\includegraphics[width=2.8in]{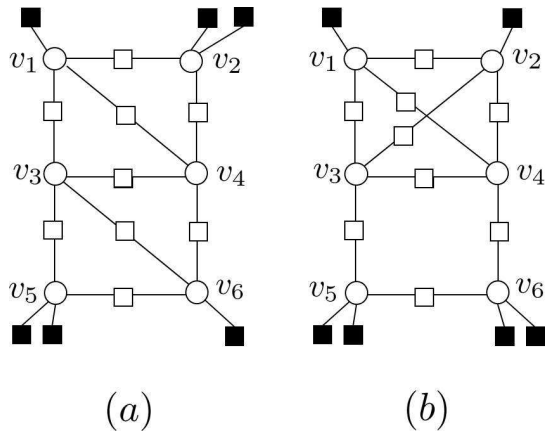}
\caption{Two possible structures of a $(6,6)$ ETS in a left-regular graph with $d_l = 4$ and $g = 6$.}
\label{66_4}
\end{figure}
\end{example}

It is however, easy to see that the structure in Figure~\ref{66_4}$(b)$ is an LSS of the set
$\{v_3, v_4, v_5,v_6\}$, which itself forms an 8-cycle. The trapping set can thus be found by Algorithm~1 if cycles of length 8 are included in the input set.
In general, adding short cycles longer than the girth to the input of Algorithm~1 can improve the coverage of this algorithm.
Nonetheless, there are structures which do not satisfy the LSS property with respect to any of their cycles. This means that even arbitrarily enlarging the
size of the cycles in the initial input set will not result in finding such structures by the algorithm.

\begin{example}
Figure~\ref{82_3} shows an $(8,2)$ ETS $\mathcal{S}$ in a left-regular graph with $d_l = 3$ and $g = 6$.
Set $\mathcal{S}$ contains cycles of length 6, 8 and 10. It is easy to check that $\mathcal{S}$ does not satisfy the LSS
property with respect to any of these cycles, and thus cannot be found using Algorithm~1 starting from any of them.
\begin{figure}[!h]
\centering
\includegraphics[width=2.8in]{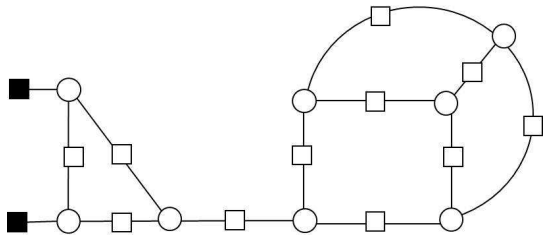}
\caption{An $(8,2)$ ETS in a left-regular graph with $d_l = 3$ and $g = 6$ which is not an LSS of any of its cycles.}
\label{82_3}
\end{figure}
\end{example}

\subsection{Non-Isomorphic Structures of Dominant ETSs and Their Characterization}

In this part of the paper, we investigate the structure of ETSs for left-regular LDPC codes with
left degrees 3, 4, 5 and 6. For each category of codes, we consider girth values 6 and 8, and
study all the non-isomorphic structures of different classes of $(a,b)$ ETSs with values of $a$ and $b$
up to 10. For each class of ETSs with given values of $a$ and $b$, we first find all the non-isomorphic
structures using the approach described in Section~III (with constraints on $d_l$ and $g$).
We then examine each of these structures to find out whether the structure is an LSS of any of its cycles.
This can be performed using Algorithm~1. Let ${\cal S}$ be the structure under consideration.
We can start with the set of shortest cycles, say of length $\ell$, in ${\cal S}$, and apply Algorithm~1 to recursively expand them to
larger subsets of ${\cal S}$. If this process ends with finding ${\cal S}$, then we report ${\cal S}$
as being an LSS of a cycle of length $\ell$. In the case that this process will not result in finding ${\cal S}$,
we use the set of cycles of next larger size as the input and repeat the process. This will continue until
${\cal S}$ is identified as an LSS of one of its cycles or until all the cycles are exhausted and ${\cal S}$ is not
an LSS of any of them. In the former case, if the cycle length is $x$, we refer to ${\cal S}$ as an LSS$_x$ structure.
The results for different values of $d_l$ and $g$ are reported
in the following subsections. For each value of $a$, we mostly consider the values of $b$ which satisfy $b/a\le 1$.
For a given value of $a$, these values of $b$ are believed to correspond to dominant trapping sets~\cite{Schlegel2010}. It is easy to see that for a Tanner graph with girth $g=8$, it is impossible to have any $(a,b)$ trapping set in $\mathcal{T}$ with $a<4$. For a Tanner graph with $g=6$ and variable node degree $d_l$, short cycles of length 6 are trivial $(3, 3(d_l-2))$ trapping sets. For this reason, we only consider $(a,b)$ trapping sets with $a\ge 4$.

\subsubsection{$d_l = 3, g = 6$}

\begin{theorem}
For left-regular graphs with $d_l = 3$ and $g = 6$, the multiplicity of non-isomorphic LSS$_x$ structures for different values
of $x$ are listed in Table ~\ref{d3g6} for different classes of ETSs in $\mathcal{T}$.
\label{thm1}
\end{theorem}

\begin{table}[!h]
\caption{\small{LSS properties of non-isomorphic structures of $(a,b)$ ETS Classes for left-regular graphs with $d_l = 3$ and $g = 6$.  }}
\centering
\begin{adjustwidth}{-1cm}{}
\renewcommand{\tabcolsep}{2pt}
\renewcommand{\arraystretch}{1.05}
\begin{tabular}{||c"c|c|c|c|c|c|c|}
\hline
\cline{1-7}
&\small{\textbf{$a=4$}}&\small{\textbf{$a=5$}}&\small{\textbf{$a=6$}}&\small{\textbf{$a=7$}}&\small{\textbf{$a=8$}}&\small{\textbf{$a=9$}} \\\hline
\cline{1-7}
\small{\textbf{$b=0$}}&\small{$\begin{Bmatrix}  g\\1  \end{Bmatrix}$}&\small{-}&\small{$\begin{Bmatrix}  g+2\\2  \end{Bmatrix}$}&\small{-}&\small{$\begin{Bmatrix}  g+4&g+6\\3&2  \end{Bmatrix}$}&\small{-}  \\\hline

\small{\textbf{$b=1$}}&\small{-}&\small{$\begin{Bmatrix}  g\\1  \end{Bmatrix}$}&\small{-}&\small{$\begin{Bmatrix}  g+2&g+4\\3&1 \end{Bmatrix}$}&\small{-}& \small{$\begin{Bmatrix}  g+4&g+6&g+8&\text{NA}\\9&7&2&1  \end{Bmatrix}$}   \\\hline

\small{\textbf{$b=2$}}&\small{$\begin{Bmatrix}  g\\1  \end{Bmatrix}$}&\small{-}&\small{$\begin{Bmatrix}  g+2&g+4\\3&1  \end{Bmatrix}$}&\small{-}&\small{$\begin{Bmatrix}  g+4&g+6&g+8&\text{NA}\\9&7&1&2  \end{Bmatrix}$}&\small{-}   \\\hline

\small{\textbf{$b=3$}}&\small{-}&\small{$\begin{Bmatrix}  g+2\\2  \end{Bmatrix}$}&\small{-}&\small{$\begin{Bmatrix}  g+4&g+6&\text{NA}\\6&3&1  \end{Bmatrix}$}&\small{-}& \small{$\begin{Bmatrix}  g+6&g+8&g+10&\text{NA}\\31&18&14&10  \end{Bmatrix}$}  \\\hline

\small{\textbf{$b=4$}}&\small{$\begin{Bmatrix}  g+2\\1  \end{Bmatrix}$}&\small{-}&\small{$\begin{Bmatrix}  g+4&g+6&\text{NA}\\2&1&1  \end{Bmatrix}$}&\small{-}&\small{$\begin{Bmatrix}  g+6&g+8&g+10&\text{NA}\\12&6&2&5  \end{Bmatrix}$}&\small{-}   \\\hline

\small{\textbf{$b=5$}}&\small{-}&\small{$\begin{Bmatrix}  g+4\\1  \end{Bmatrix}$}&\small{-}&\small{$\begin{Bmatrix}  g+6&g+8&\text{NA}\\3&1&2 \end{Bmatrix}$}&\small{-}& \small{$\begin{Bmatrix}  g+8&g+10&g+12&\text{NA}\\19&13&3&17  \end{Bmatrix}$}  \\\hline

\small{\textbf{$b=6$}}&\small{-}&\small{-}&\small{$\begin{Bmatrix}  g+6\\1  \end{Bmatrix}$}&\small{-}&\small{$\begin{Bmatrix}  g+8&\text{NA}\\3&7  \end{Bmatrix}$}& \small{-}  \\\hline

\small{\textbf{$b=7$}}&\small{-}&\small{-}&\small{-}&\small{$\begin{Bmatrix}  g+8\\1  \end{Bmatrix}$}&\small{-}& \small{$\begin{Bmatrix}  g+10&g+12&\text{NA}\\4&2&7  \end{Bmatrix}$}  \\\hline

\small{\textbf{$b=8$}}&\small{-}&\small{-}&\small{  -  }&\small{-}&\small{$\begin{Bmatrix}  \text{$g+10$}\\1  \end{Bmatrix}$}& \small{-}  \\\hline

\cline{1-7}
\end{tabular}
\end{adjustwidth}
\label{d3g6}
\end{table}

(Each row in Table~\ref{d3g6} corresponds to a specific size $a$ of an ETS, and each column corresponds to a specific
number of unsatisfied check nodes, $b$. For each pair $(a,b)$, Table~\ref{d3g6} lists the multiplicity of LSS$_x$ structures
for different values of $x$. For example, having ${\tiny \begin{Bmatrix}  g+2&g+4\\3&1  \end{Bmatrix}}$ for a specific
class of ETSs means that there are 4 different non-isomorphic structures for that class of trapping sets:
Three of them are LSS$_{g+2}$ and one of them is an LSS$_{g+4}$ structure. Starting with cycles of length $g+2$ and $g+4$,
Algorithm~1 is thus guaranteed to find all such ETSs.  Having the symbol ``-" for a class of $(a,b)$ trapping sets means
that for the underlying conditions, i.e., $d_l = 3$ and $g = 6$, it is impossible to have such a class of trapping sets.
For the structures which do not satisfy the LSS property with respect to any of their cycles, we use the notation ``NA" (stands for not applicable).
Starting from any set of cycles of the graph, Algorithm~1 cannot find such structures.)

\begin{proof}
In general, the results reported in Table~\ref{d3g6} for any $(a,b)$ ETS can be proved by first obtaining
all the non-isomorphic structures using the nauty program, as described in Section~III, and then examining
each such structure for the LSS property, using Algorithm~1, as described earlier in this subsection. We however, provide
a formal proof for the results pertaining to the ETS class $(6,0)$ as well as all the classes that cannot exist
in left-regular graphs with $d_l = 3$ and $g = 6$ (i.e., those with designation ``-" in the table). The general approach
for the formal proof of the rest of the results is similar and not provided.

\underline{$(6,0)$ ETSs}: We first prove that the two structures presented in Figure~\ref{6_0_NG} are the only possible non-isomorphic
structures for $(6, 0)$ ETSs in left-regular graphs with $d_l = 3$ and $g = 6$. We use the normal graph representation of these structures,
also shown in  Figure~\ref{6_0_NG}, for the proof.

\begin{figure}[!h]
\centering
\includegraphics[width=3.7in]{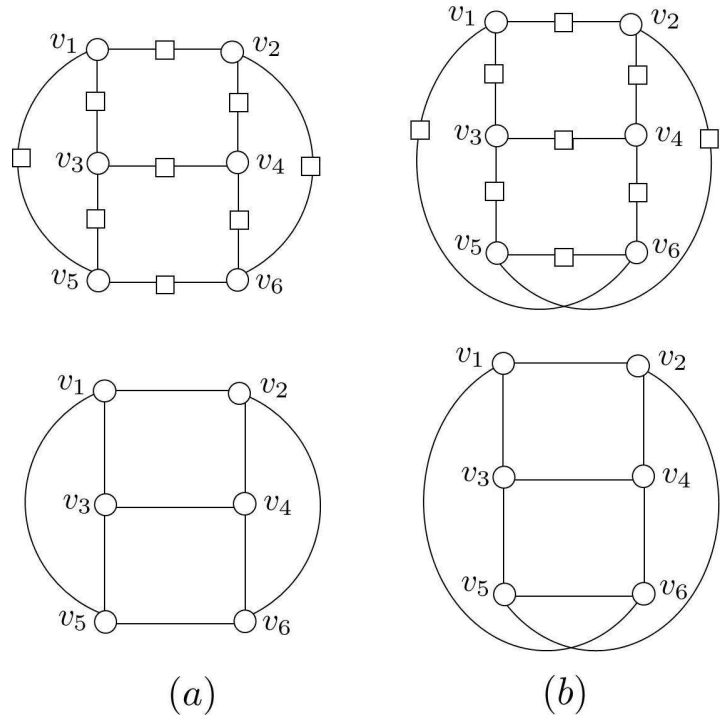}
\caption{The only possible non-isomorphic structures for $(6, 0)$ ETSs in left-regular graphs with $d_l = 3$ and $g = 6$, and their normal graph representations.}
\label{6_0_NG}
\end{figure}

Since the number of unsatisfied check nodes is zero, every node in the normal graph must be connected to three other nodes.
Starting from an arbitrary node, say $v_3$, as the root, we grow the normal graph, and will have three nodes in the first layer.
We arbitrarily denote these nodes by $v_1$, $v_4$, and $v_5$. Consider the case where the nodes in the first layer
do not have any edges in common (see Figure~\ref{ab_grow1}$(a)$). In this case, they must have 6 edges connected to the two remaining nodes
($v_2$ and $v_6$). According to the girth constraint (i.e., $g=6$), each pair of nodes can have at most one edge in common.
This implies that the only possible scenario is the case where each of the three nodes in the first layer is connected
to both $v_2$ and $v_6$. This results in the structure of Figure~\ref{6_0_NG}$(b)$. Note that in this case,
the length of the shortest cycles in the structure is 8.

Now, consider the case where the nodes in the first layer have some edge(s) in common. Based on the girth constraint ($g = 6$),
the two remaining nodes ($v_2$ and $v_6$) must have  at least 4 edges in common with the nodes in the first layer. (Otherwise,
for $v_2$ and $v_6$ to have degree 3, they need to have more than one edge in common, which contradicts the girth constraint.)
This implies that the nodes in the first layer can have only one edge in common. Without loss of generality,
we assume $v_1$ and $v_4$ are connected (see Figure~\ref{ab_grow1}$(b)$).
From this point on, there is no choice in connecting the nodes. Node $v_5$ must be connected to both $v_2$ and $v_6$,
and each of the nodes $v_1$ and $v_4$ must have one connection to $v_2$ or $v_6$. (Switching the connections will result in
isomorphic structures.) This results in the structure of Figure~\ref{6_0_NG}$(a)$.

\begin{figure}[!h]
\centering
\includegraphics[width=3in]{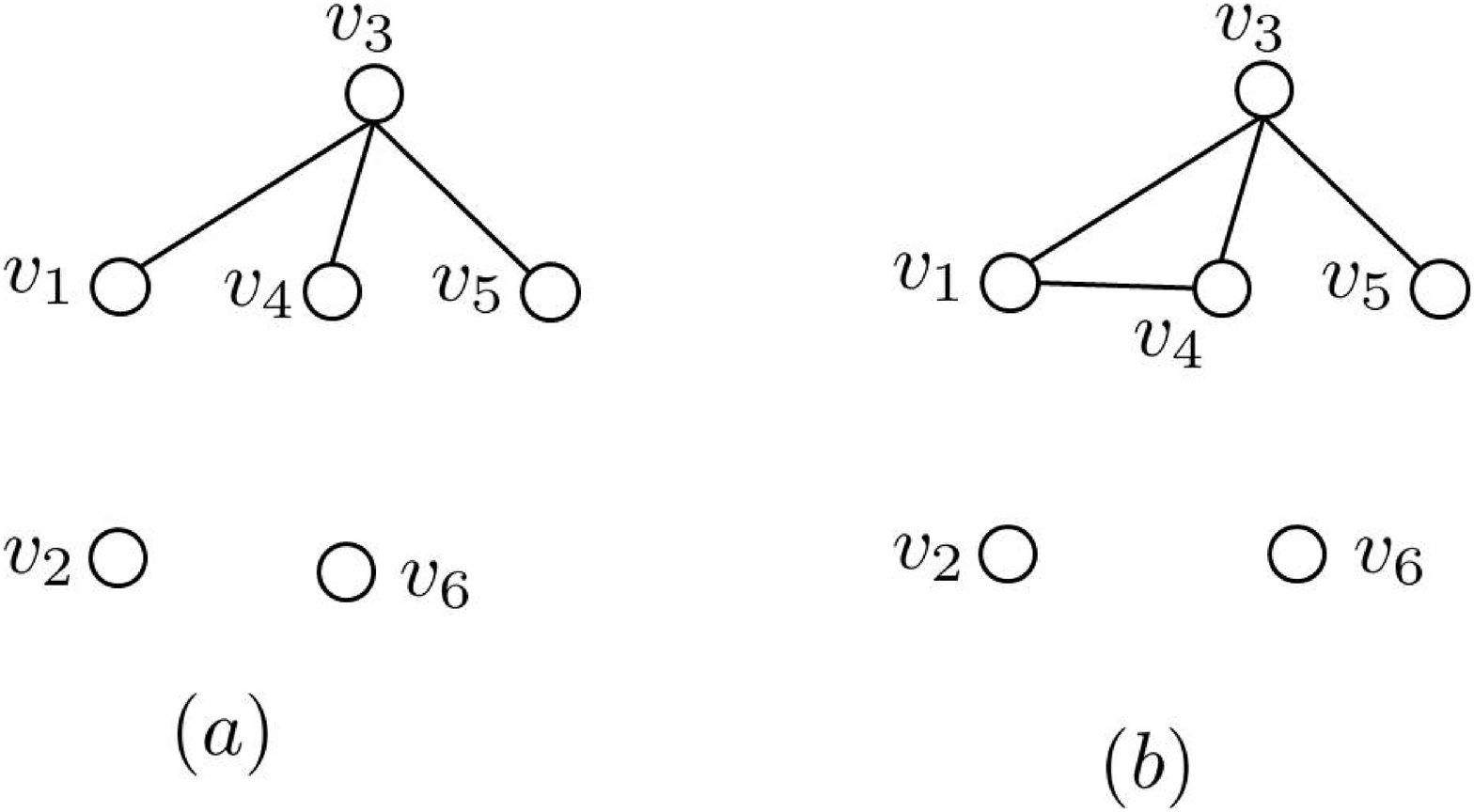}
\caption{Growing the normal graph of $(6,0)$ ETSs of a left-regular graph with $d_l = 3$ and $g = 6$.}
\label{ab_grow1}
\end{figure}

There are two cycles of length 6 in the structure $\mathcal{S}_1$ of Figure~\ref{6_0_NG}$(a)$. However, $\mathcal{S}_1$ is not an LSS of any of them.
There are also three cycles of length 8 in $\mathcal{S}_1$, and $\mathcal{S}_1$ is an LSS of all of them. For example, $\mathcal{S}_1$
is an LSS of $C_1=\{v_1,v_3,v_4, v_2\}$ with the following nested sequence of ETSs: $C_1 \subset \{v_1,v_3,v_4, v_2, v_5\} \subset \mathcal{S}_1$.
The structure $\mathcal{S}_2$ of Figure ~\ref{6_0_NG}$(b)$ does not have any 6-cycles. It has four 8-cycles and has LSS property with respect to all of them.

\underline{Non-existent structures}: All such cases in Table~\ref{d3g6} can be proved using two facts: (i) An $(a,b)$ ETS in a left-regular graph with $d_l = 3$
has $3a$ edges and $(3a-b)/2$ satisfied check nodes. For the latter to be integer, there cannot be any $(a,b)$ ETS with an odd $a$ and an even $b$,
or vice versa. (ii) By definition, each variable node of an ETS ${\cal S} \in {\cal T}$ of a left-regular graph with $d_l = 3$
can be connected to at most one unsatisfied check node. This means that it is not possible to have an $(a,b)$ ETS with $b>a$.
\end{proof}

The results of Table~\ref{d3g6} indicate that the majority of ETS structures satisfy the LSS property with respect to
short cycles. This is particularly the case for smaller values of $a$ and $b$ which correspond to more problematic trapping sets.

\begin{example}
A cycle of length $2a$ in a left-regular graph with $d_l = 3$ is an $(a,a)$ ETS. For example, Figure~\ref{55_3_n} shows the structure
of a $(5,5)$ ETS in a left-regular graph with $d_l = 3$, which is a cycle of length 10. In fact, as the results of Table~\ref{d3g6} show
cycles of length $2a$ are the only $(a,a)$ ETSs for the graphs under consideration.
Such cycles trivially satisfy the LSS property.

\begin{figure}[!h]
\centering
\includegraphics[width=1.5in]{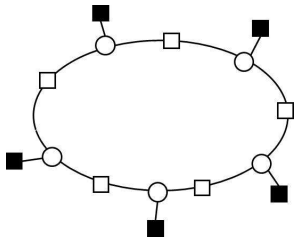}
\caption{A cycle of length 10 is the only possible structure for a $(5,5)$ ETS in a left-regular graph with $d_l = 3$ and $g = 6$.}
\label{55_3_n}
\end{figure}

\end{example}

\begin{example}
In Figure~\ref{N64_73_91_3_1}, the three structures corresponding to the ``NA'' cases in Table~\ref{d3g6} for $(6,4)$, $(7,3)$ and $(9,1)$ ETSs
are presented. All three structures consist of two ETSs connected by a check node.
\begin{figure}[!h]
\centering
\includegraphics[width=4.5in]{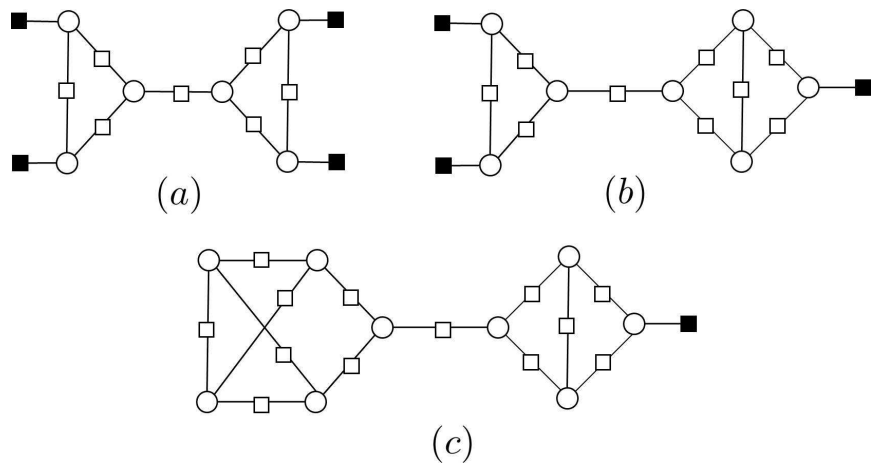}
\caption{The three structures corresponding to the ``NA'' cases in Table~\ref{d3g6} for $(6,4)$, $(7,3)$ and $(9,1)$ ETSs.}
\label{N64_73_91_3_1}
\end{figure}

\end{example}

In the rest of the paper, we omit the proofs as they are similar in nature to that of Theorem~\ref{thm1}.

\subsubsection{$d_l = 3, g = 8$}

\begin{theorem}
For left-regular graphs with $d_l = 3$ and $g = 8$, the multiplicity of non-isomorphic LSS$_x$ structures for different values
of $x$ are listed in Table ~\ref{d3g8} for different classes of ETSs in $\mathcal{T}$.
\end{theorem}

\begin{table}[!h]
\caption{\small{LSS properties of non-isomorphic structures of $(a,b)$ ETS Classes for left-regular graphs with $d_l = 3$ and $g = 8$.   }}
\centering
\begin{tabular}{||c"c|c|c|c|c|c|c|}
\hline
\cline{1-7}
&\small{\textbf{$a=4$}}&\small{\textbf{$a=5$}}&\small{\textbf{$a=6$}}&\small{\textbf{$a=7$}}&\small{\textbf{$a=8$}}&\small{\textbf{$a=9$}} \\\hline
\cline{1-7}
\small{\textbf{$b=0$}}&\small{-}&\small{-}&\small{$\begin{Bmatrix}  g\\1  \end{Bmatrix}$}&\small{-}&\small{$\begin{Bmatrix}  g+2&g+4\\1&1  \end{Bmatrix}$}&\small{-}  \\\hline

\small{\textbf{$b=1$}}&\small{-}&\small{-}&\small{-}&\small{$\begin{Bmatrix}  g\\1  \end{Bmatrix}$}&\small{-}& \small{$\begin{Bmatrix}  g+2&g+4\\3&1  \end{Bmatrix}$}   \\\hline

\small{\textbf{$b=2$}}&\small{-}&\small{-}&\small{$\begin{Bmatrix}  g\\1  \end{Bmatrix}$}&\small{-}&\small{$\begin{Bmatrix}  g+2&g+4 \\3&2 \end{Bmatrix}$}&\small{-}   \\\hline

\small{\textbf{$b=3$}}&\small{-}&\small{$\begin{Bmatrix}  g\\1  \end{Bmatrix}$}&\small{-}&\small{$\begin{Bmatrix}  g+2&g+4\\2&1  \end{Bmatrix}$}&\small{-}& \small{$\begin{Bmatrix}  g+4&g+6\\13&4  \end{Bmatrix}$}  \\\hline

\small{\textbf{$b=4$}}&\small{$\begin{Bmatrix}  g\\1  \end{Bmatrix}$}&\small{-}&\small{$\begin{Bmatrix}  g+2&g+4\\1&1  \end{Bmatrix}$}&\small{-}&\small{$\begin{Bmatrix}  g+4&g+6&g+8\\6&2&2  \end{Bmatrix}$}&\small{-}   \\\hline

\small{\textbf{$b=5$}}&\small{-}&\small{$\begin{Bmatrix}  g+2\\1  \end{Bmatrix}$}&\small{-}&\small{$\begin{Bmatrix}  g+4&g+6\\2&1  \end{Bmatrix}$}&\small{-}& \small{$\begin{Bmatrix}  g+6&g+8&g+10&\text{NA}\\10&7&3&1  \end{Bmatrix}$}  \\\hline

\small{\textbf{$b=6$}}&\small{-}&\small{-}&\small{$\begin{Bmatrix}  g+4\\1  \end{Bmatrix}$}&\small{-}&\small{$\begin{Bmatrix}  g+6&g+8&\text{NA}\\2&2&2  \end{Bmatrix}$}& \small{-}  \\\hline

\small{\textbf{$b=7$}}&\small{-}&\small{-}&\small{-}&\small{$\begin{Bmatrix}  g+6\\1  \end{Bmatrix}$}&\small{-}& \small{$\begin{Bmatrix} g+8&g+10& \text{NA}\\3&2&3  \end{Bmatrix}$}  \\\hline

\small{\textbf{$b=8$}}&\small{-}&\small{-}&\small{-}&\small{-}&\small{$\begin{Bmatrix}  g+8\\1  \end{Bmatrix}$}& \small{-} \\\hline

\cline{1-7}
\end{tabular}
\label{d3g8}
\end{table}

The results of Table~\ref{d3g8} also indicate that the majority of ETS structures satisfy the LSS property with respect to short cycles of the graph.
Compared to the results of Table~\ref{d3g6}, a larger portion of ETS classes have the property that all their non-isomorphic structures
are LSSs of short cycles.

\begin{remark}
In a left-regular graph with $d_l =3$, any ETS in $\mathcal{T}$ is also an absorbing set. So the results presented in Tables~\ref{d3g6} and~\ref{d3g8}
are also applicable to the absorbing sets of these graphs.
\end{remark}

\subsubsection{$d_l = 4$ and $g = 6$}

\begin{theorem}
For left-regular graphs with $d_l = 4$ and $g = 6$, the multiplicity of non-isomorphic LSS$_x$ structures for different values
of $x$ are listed in Tables~\ref{d4g6_1} and~\ref{d4g6_2} for different classes of ETSs in $\mathcal{T}$.
\end{theorem}

\begin{table}[!h]
\caption{\small{LSS properties of non-isomorphic structures of $(a,b)$ ETS Classes for left-regular graphs with $d_l = 4$ and $g = 6$.}}
\centering
\begin{tabular}{||c"c|c|c|c|c|c|c|}
\hline
\cline{1-5}
&\small{\textbf{$a=4$}}&\small{\textbf{$a=5$}}&\small{\textbf{$a=6$}}&\small{\textbf{$a=7$}} \\\hline
\cline{1-5}
\small{\textbf{$b=0$}}&\small{-}&\small{TS:$\begin{Bmatrix}  g\\1  \end{Bmatrix}$}&\small{TS:$\begin{Bmatrix}  g \\1 \end{Bmatrix}$}&\small{TS:$\begin{Bmatrix}  g\\2  \end{Bmatrix}$}  \\
&\small{-}&\small{AS:$\begin{Bmatrix}  g\\1  \end{Bmatrix}$}&\small{AS:$\begin{Bmatrix}  g\\1  \end{Bmatrix}$}&\small{AS:$\begin{Bmatrix}  g\\2  \end{Bmatrix}$}  \\\hline

\small{\textbf{$b=1$}}&\small{-}&\small{-}&\small{-}&\small{-}  \\\hline
\small{\textbf{$b=2$}}&\small{-}&\small{TS:$\begin{Bmatrix}  g\\1  \end{Bmatrix}$}&\small{TS:$\begin{Bmatrix}  g\\3  \end{Bmatrix}$}&\small{TS:$\begin{Bmatrix}  g\\9  \end{Bmatrix}$}   \\
&\small{-}&\small{AS:$\begin{Bmatrix}  g\\1  \end{Bmatrix}$}&\small{AS:$\begin{Bmatrix}  g \\2 \end{Bmatrix}$}&\small{AS:$\begin{Bmatrix}  g\\7  \end{Bmatrix}$}   \\\hline

\small{\textbf{$b=3$}}&\small{-}&\small{-}&\small{-}&\small{-}  \\\hline
\small{\textbf{$b=4$}}&\small{TS:$\begin{Bmatrix}  g\\1 \end{Bmatrix}$}&\small{TS:$\begin{Bmatrix}  g\\2  \end{Bmatrix}$}&\small{TS:$\begin{Bmatrix}  g \\7 \end{Bmatrix}$}&\small{TS:$\begin{Bmatrix}  g&g+2&g+4\\25&2&1 \end{Bmatrix}$}  \\
&\small{AS:$\begin{Bmatrix}  g\\1 \end{Bmatrix}$}&\small{AS:$\begin{Bmatrix}  g \\1 \end{Bmatrix}$}&\small{AS:$\begin{Bmatrix}  g\\3  \end{Bmatrix}$}&\small{AS:$\begin{Bmatrix}  g&g+2\\9&2  \end{Bmatrix}$}   \\\hline

\small{\textbf{$b=5$}}&\small{-}&\small{-}&\small{-}&\small{-}  \\\hline
\small{\textbf{$b=6$}}&\small{TS:$\begin{Bmatrix}  g\\1 \end{Bmatrix}$}&\small{TS:$\begin{Bmatrix}  g \\3 \end{Bmatrix}$}&\small{TS:$\begin{Bmatrix}  g&g+2\\8&3  \end{Bmatrix}$}&\small{TS:$\begin{Bmatrix}  g&g+2&g+4&\text{NA}\\28&12&3&1 \end{Bmatrix}$} \\
&\small{AS: -}&\small{AS: -}&\small{AS:$\begin{Bmatrix}  g+2\\2  \end{Bmatrix}$}&\small{AS:$\begin{Bmatrix}  g+2&g+4\\3&1  \end{Bmatrix}$}  \\\hline

\small{\textbf{$b=7$}}&\small{-}&\small{-}&\small{-}&\small{-}  \\\hline
\small{\textbf{$b=8$}}&\small{TS:$\begin{Bmatrix}  g+2\\1 \end{Bmatrix}$}&\small{TS:$\begin{Bmatrix}  g+2 &\text{NA}\\2&1 \end{Bmatrix}$}&\small{TS:$\begin{Bmatrix}  g+2&g+4&\text{NA}\\8&1&1  \end{Bmatrix}$}&\small{TS:$\begin{Bmatrix}  g+2&g+4&g+6&\text{NA}\\29&9&1&5  \end{Bmatrix}$} \\
&\small{AS: -}&\small{AS: -}&\small{AS: -}&\small{AS: -} \\\hline

\cline{1-5}
\end{tabular}
\label{d4g6_1}
\end{table}

\begin{table}[!h]
\caption{\small{LSS properties of non-isomorphic structures of $(a,b)$ ETS Classes for left-regular graphs with $d_l = 4$ and $g = 6$.  }}
\centering
\begin{tabular}{||c"c|c|c|c|c|c|c|}
\hline
\cline{1-3}
&\small{\textbf{$a=8$}}&\small{\textbf{$a=9$}} \\\hline
\cline{1-3}
\small{\textbf{$b=0$}}&\small{TS:$\begin{Bmatrix}  g&g+2\\4&2   \end{Bmatrix}$}&\small{TS:$\begin{Bmatrix}  g&g+2\\10&6  \end{Bmatrix}$}  \\
&\small{AS:$\begin{Bmatrix}  g&g+2\\4&2  \end{Bmatrix}$}&\small{AS:$\begin{Bmatrix}  g&g+2\\10&6  \end{Bmatrix}$}  \\\hline

\small{\textbf{$b=1$}}&\small{-}&\small{-}  \\\hline
\small{\textbf{$b=2$}}&\small{TS:$\begin{Bmatrix}  g&g+2\\32&3  \end{Bmatrix}$}&\small{TS:$\begin{Bmatrix}  g&g+2&g+4\\127&24&3  \end{Bmatrix}$}   \\
&\small{AS:$\begin{Bmatrix}  g&g+2\\25&3  \end{Bmatrix}$}&\small{AS:$\begin{Bmatrix} g&g+2&g+4\\102&21&3  \end{Bmatrix}$}   \\\hline

\small{\textbf{$b=3$}}&\small{-}&\small{-}  \\\hline
\small{\textbf{$b=4$}}&\small{TS:$\begin{Bmatrix}  g&g+2&g+4&g+6&\text{NA}\\101&18&3&1&1  \end{Bmatrix}$}&\small{TS:$\begin{Bmatrix}  g&g+2&g+4&g+6&\text{NA}\\460&165&26&7&5 \end{Bmatrix}$}   \\
&\small{AS:$\begin{Bmatrix}  g&g+2&g+4\\34&15&1  \end{Bmatrix}$}&\small{AS:$\begin{Bmatrix}  g&g+2&g+4&g+6&\text{NA}\\154&110&16&3&2  \end{Bmatrix}$}   \\\hline

\small{\textbf{$b=5$}}&\small{-}&\small{-}  \\\hline
\small{\textbf{$b=6$}}&\small{TS:$\begin{Bmatrix}  g&g+2&g+4&g+6&\text{NA}\\116&81&21&6&7  \end{Bmatrix}$}& \small{TS:$\begin{Bmatrix}  g&g+2&g+4&g+6&g+8&\text{NA}\\523&617&149&51&6&33  \end{Bmatrix}$}  \\
&\small{AS:$\begin{Bmatrix}  g+2&g+4&g+6&\text{NA}\\22&4&1&1  \end{Bmatrix}$}& \small{AS:$\begin{Bmatrix}  g+2&g+4&g+6&g+8&\text{NA}\\131&32&10&1&3  \end{Bmatrix}$}  \\\hline

\small{\textbf{$b=7$}}&\small{-}&\small{-}  \\\hline
\small{\textbf{$b=8$}}&\small{TS:$\begin{Bmatrix}  g+2&g+4&g+6&g+8&g+10&\text{NA}\\144&63&21&1&1&20 \end{Bmatrix}$}& \small{TS:$\begin{Bmatrix}  g+2&g+4&g+6&g+8&g+10&\text{NA}\\855&446&173&30&5&104  \end{Bmatrix}$} \\
&\small{AS:$\begin{Bmatrix} g+4&g+6\\3&2  \end{Bmatrix}$}& \small{AS:$\begin{Bmatrix}  g+4&g+6&g+8&\text{NA}\\18&6&1&2  \end{Bmatrix}$} \\\hline

\cline{1-3}
\end{tabular}
\label{d4g6_2}
\end{table}

For graphs with $d_l > 3$, not every ETS is an absorbing set. To identify the structures that are absorbing sets, we use
the notation ``AS'' in the tables. For each class of ETS, we thus have two sets of entries: the ones for absorbing sets
denoted by ``AS'', and those that correspond to all the ETSs denoted by ``TS'' to stand for ``trapping sets''.
For each entry, the number of possible non-absorbing set structures can be obtained by subtracting the corresponding results in the two sets.
For example, the results for the ETS class $(5,4)$ in Table~\ref{d4g6_1} indicate that there are two non-isomorphic structures for this class that are both
LSSs of cycles of length $g$. Only one of the two structures however, is an absorbing set.

Tables~\ref{d4g6_1} and~\ref{d4g6_2} indicate that an overwhelming majority of ETS structures satisfy the LSS property
with respect to short cycles. In particular, all the ETSs of size less than 7 with less than 6 unsatisfied check nodes
are LSSs of 6-cycles.

\subsubsection{$d_l = 4$ and $g = 8$}

\begin{theorem}
For left-regular graphs with $d_l = 4$ and $g = 8$, the multiplicity of non-isomorphic LSS$_x$ structures for different values
of $x$ are listed in Table~\ref{d4g8} for different classes of ETSs in $\mathcal{T}$.
\end{theorem}

\begin{table}[!h]
\caption{\small{LSS properties of non-isomorphic structures of $(a,b)$ ETS Classes for left-regular graphs with $d_l = 4$ and $g = 8$.  }}
\centering
\begin{tabular}{||c"c|c|c|c|c|c|c|}
\hline
\cline{1-7}
&\small{\textbf{$a=4$}}&\small{\textbf{$a=5$}}&\small{\textbf{$a=6$}}&\small{\textbf{$a=7$}}&\small{\textbf{$a=8$}}&\small{\textbf{$a=9$}} \\\hline
\cline{1-7}
\small{\textbf{$b=0$}}&\small{-}&\small{-}&\small{-}&\small{-}&\small{TS: $\begin{Bmatrix}  g \\1 \end{Bmatrix}$}&\small{TS: -}  \\
&\small{-}&\small{-}&\small{-}&\small{-}&\small{AS: $\begin{Bmatrix}  g \\1 \end{Bmatrix}$}&\small{AS: -}  \\\hline

\small{\textbf{$b=1$}}&\small{-}&\small{-}&\small{-}&\small{-}&\small{-}&\small{-}  \\\hline
\small{\textbf{$b=2$}}&\small{-}&\small{-}&\small{-}&\small{-}&\small{TS: $\begin{Bmatrix}  g\\1  \end{Bmatrix}$}&\small{TS: $\begin{Bmatrix}  g \\2 \end{Bmatrix}$}   \\
&\small{-}&\small{-}&\small{-}&\small{-}&\small{AS: $\begin{Bmatrix}  g\\1  \end{Bmatrix}$}&\small{AS: $\begin{Bmatrix}  g \\1 \end{Bmatrix}$}   \\\hline

\small{\textbf{$b=3$}}&\small{-}&\small{-}&\small{-}&\small{-}&\small{-}&\small{-}  \\\hline
\small{\textbf{$b=4$}}&\small{-}&\small{-}&\small{-}&\small{TS: $\begin{Bmatrix}  g \\1 \end{Bmatrix}$}&\small{TS: $\begin{Bmatrix}  g\\2  \end{Bmatrix}$}&\small{TS: $\begin{Bmatrix}  g\\7  \end{Bmatrix}$}   \\
&\small{-}&\small{-}&\small{-}&\small{AS: $\begin{Bmatrix}  g\\1  \end{Bmatrix}$}&\small{AS: $\begin{Bmatrix}  g \\1 \end{Bmatrix}$}&\small{AS: $\begin{Bmatrix}  g \\3 \end{Bmatrix}$}   \\\hline

\small{\textbf{$b=5$}}&\small{-}&\small{-}&\small{-}&\small{-}&\small{-}&\small{-}  \\\hline
\small{\textbf{$b=6$}}&\small{-}&\small{-}&\small{TS: $\begin{Bmatrix}  g\\1  \end{Bmatrix}$}&\small{TS: $\begin{Bmatrix}  g \\1 \end{Bmatrix}$}&\small{TS: $\begin{Bmatrix}  g \\5 \end{Bmatrix}$}& \small{TS: $\begin{Bmatrix}  g&g+2\\18&1  \end{Bmatrix}$}  \\
&\small{-}&\small{-}&\small{AS: $\begin{Bmatrix}  g\\1  \end{Bmatrix}$}&\small{AS: -}&\small{AS: $\begin{Bmatrix}  g \\2 \end{Bmatrix}$}& \small{AS: $\begin{Bmatrix}  g\\5  \end{Bmatrix}$}  \\\hline

\small{\textbf{$b=7$}}&\small{-}&\small{-}&\small{-}&\small{-}&\small{-}&\small{-}  \\\hline
\small{\textbf{$b=8$}}&\small{TS: $\begin{Bmatrix}  g\\1  \end{Bmatrix}$}&\small{TS: $\begin{Bmatrix}  g\\1  \end{Bmatrix}$}&\small{TS: $\begin{Bmatrix}  g \\2 \end{Bmatrix}$}&\small{TS: $\begin{Bmatrix}  g\\3  \end{Bmatrix}$}&\small{TS: $\begin{Bmatrix}  g&g+2&g+4\\10&2&2  \end{Bmatrix}$}& \small{TS: $\begin{Bmatrix}  g&g+2&g+4 \\36&10&4 \end{Bmatrix}$} \\
&\small{AS: -}&\small{AS: -}&\small{AS: -}&\small{AS: -}&\small{AS: $\begin{Bmatrix}  g+2&g+4\\1&1  \end{Bmatrix}$}& \small{AS: $\begin{Bmatrix}  g\\3  \end{Bmatrix}$} \\\hline

\cline{1-7}
\end{tabular}
\label{d4g8}
\end{table}

The results of Table~\ref{d4g8} indicate that all the ETSs of size less than 10 with less than 9 unsatisfied check nodes satisfy the LSS property.
Comparison with the results of Tables~\ref{d4g6_1} and \ref{d4g6_2} reveals that by increasing the girth of the graph from 6 to 8,
the number of classes whose structures all satisfy the LSS property increases.

\subsubsection{$d_l = 5$ and $g = 6$}

\begin{theorem}
For left-regular graphs with $d_l = 5$ and $g = 6$, the multiplicity of non-isomorphic LSS$_x$ structures for different values
of $x$ are listed in Table~\ref{d5g6} for different classes of ETSs in $\mathcal{T}$.
\end{theorem}

\begin{table}[!h]
\caption{\small{LSS properties of non-isomorphic structures of $(a,b)$ ETS Classes for left-regular graphs with $d_l = 5$ and $g = 6$.   }}
\centering
\begin{adjustwidth}{-.6cm}{}
\renewcommand{\tabcolsep}{2pt}
\begin{tabular}{||c"c|c|c|c|c|c|c|}
\hline
\cline{1-7}
&\small{\textbf{$a=4$}}&\small{\textbf{$a=5$}}&\small{\textbf{$a=6$}}&\small{\textbf{$a=7$}}&\small{\textbf{$a=8$}}&\small{\textbf{$a=9$}} \\\hline
\cline{1-7}
\small{\textbf{$b=0$}}&\small{-}&\small{-}&\small{TS:$\begin{Bmatrix}  g \\1 \end{Bmatrix}$}&\small{-}&\small{TS:$\begin{Bmatrix}  g \\3 \end{Bmatrix}$}&\small{-}  \\
&\small{-}&\small{-}&\small{AS:$\begin{Bmatrix}  g \\1 \end{Bmatrix}$}&\small{-}&\small{AS:$\begin{Bmatrix}  g \\3 \end{Bmatrix}$}&\small{-}  \\\hline

\small{\textbf{$b=1$}}&\small{-}&\small{-}&\small{-}&\small{TS:$\begin{Bmatrix}  g\\1  \end{Bmatrix}$}&\small{-}& \small{TS:$\begin{Bmatrix}  g \\28 \end{Bmatrix}$}   \\
&\small{-}&\small{-}&\small{-}&\small{AS:$\begin{Bmatrix}  g\\1  \end{Bmatrix}$}&\small{-}& \small{AS:$\begin{Bmatrix}  g \\28 \end{Bmatrix}$}   \\\hline

\small{\textbf{$b=2$}}&\small{-}&\small{-}&\small{TS:$\begin{Bmatrix}  g\\1  \end{Bmatrix}$}&\small{-}&\small{TS:$\begin{Bmatrix}  g\\16  \end{Bmatrix}$}&\small{-}   \\
&\small{-}&\small{-}&\small{AS:$\begin{Bmatrix}  g\\1  \end{Bmatrix}$}&\small{-}&\small{AS:$\begin{Bmatrix}  g \\16 \end{Bmatrix}$}&\small{-}   \\\hline

\small{\textbf{$b=3$}}&\small{-}&\small{-}&\small{-}&\small{TS:$\begin{Bmatrix}  g\\6  \end{Bmatrix}$}&\small{-}& \small{TS:$\begin{Bmatrix}  g \\289 \end{Bmatrix}$}  \\
&\small{-}&\small{-}&\small{-}&\small{AS:$\begin{Bmatrix}  g \\5 \end{Bmatrix}$}&\small{-}& \small{AS:$\begin{Bmatrix}  g \\276 \end{Bmatrix}$}  \\\hline

\small{\textbf{$b=4$}}&\small{-}&\small{-}&\small{TS:$\begin{Bmatrix}  g \\2 \end{Bmatrix}$}&\small{-}&\small{TS:$\begin{Bmatrix}  g\\75  \end{Bmatrix}$}&\small{-}   \\
&\small{-}&\small{-}&\small{AS:$\begin{Bmatrix}  g\\2  \end{Bmatrix}$}&\small{-}&\small{AS:$\begin{Bmatrix}  g \\68 \end{Bmatrix}$}&\small{-}   \\\hline

\small{\textbf{$b=5$}}&\small{-}&\small{TS:$\begin{Bmatrix}  g\\1  \end{Bmatrix}$}&\small{-}&\small{TS:$\begin{Bmatrix}  g \\18 \end{Bmatrix}$}&\small{-}& \small{TS:$\begin{Bmatrix}  g&g+2 \\1355&2  \end{Bmatrix}$} \\
&\small{-}&\small{AS:$\begin{Bmatrix}  g\\1  \end{Bmatrix}$}&\small{-}&\small{AS:$\begin{Bmatrix}  g \\14 \end{Bmatrix}$}&\small{-}& \small{AS:$\begin{Bmatrix}  g&g+2 \\1149&2 \end{Bmatrix}$} \\\hline

\small{\textbf{$b=6$}}&\small{-}&\small{-}&\small{TS:$\begin{Bmatrix}  g \\5 \end{Bmatrix}$}&\small{-}&\small{TS:$\begin{Bmatrix}  g&g+4 \\222&1 \end{Bmatrix}$}& \small{-}  \\
&\small{-}&\small{-}&\small{AS:$\begin{Bmatrix}  g\\4  \end{Bmatrix}$}&\small{-}&\small{AS:$\begin{Bmatrix}  g\\165  \end{Bmatrix}$}& \small{-}  \\\hline

\small{\textbf{$b=7$}}&\small{-}&\small{TS:$\begin{Bmatrix}  g \\1 \end{Bmatrix}$}&\small{-}&\small{TS:$\begin{Bmatrix}  g \\37 \end{Bmatrix}$}&\small{-}& \small{TS:$\begin{Bmatrix}  g&g+2&g+4&g+6&\text{NA} \\3768&9&6&1&3  \end{Bmatrix}$}  \\
&\small{-}&\small{AS:$\begin{Bmatrix}  g \\1 \end{Bmatrix}$}&\small{-}&\small{AS:$\begin{Bmatrix}  g\\23  \end{Bmatrix}$}&\small{-}& \small{AS:$\begin{Bmatrix}  g&g+2&g+4\\2533&7&1 \end{Bmatrix}$}  \\\hline

\small{\textbf{$b=8$}}&\small{TS:$\begin{Bmatrix}  g \\1 \end{Bmatrix}$}&\small{-}&\small{TS:$\begin{Bmatrix}  g \\8 \end{Bmatrix}$}&\small{-}&\small{TS:$\begin{Bmatrix}  g&g+2&g+4& \text{NA} \\ 453&5&2&1 \end{Bmatrix}$}& \small{-} \\
&\small{AS:$\begin{Bmatrix}  g \\1 \end{Bmatrix}$}&\small{-}&\small{AS:$\begin{Bmatrix}  g\\5  \end{Bmatrix}$}&\small{-}&\small{AS:$\begin{Bmatrix}  g&g+2 \\249&3  \end{Bmatrix}$}& \small{-} \\\hline

\small{\textbf{$b=9$}}&\small{-}&\small{TS:$\begin{Bmatrix}  g \\2 \end{Bmatrix}$}&\small{-}&\small{TS:$\begin{Bmatrix}  g&g+2 \\61&1 \end{Bmatrix}$}&\small{-}& \small{TS:$\begin{Bmatrix}  g&g+2&g+4&g+6&\text{NA} \\ 6957&66&43&7&19  \end{Bmatrix}$}  \\
&\small{-}&\small{AS:$\begin{Bmatrix}  g\\1  \end{Bmatrix}$}&\small{-}&\small{AS:$\begin{Bmatrix}  g\\25  \end{Bmatrix}$}&\small{-}& \small{AS:$\begin{Bmatrix}  g&g+2&g+4&\text{NA} \\3243&33&7&1  \end{Bmatrix}$}  \\\hline

\cline{1-7}
\end{tabular}
\end{adjustwidth}

\label{d5g6}
\end{table}

Table~\ref{d5g6} shows that except for a small fraction of ETS structures, all the rest of the structures of ETS classes with size less than 10 and with less than 10 unsatisfied
check nodes satisfy the LSS property. In particular, all  the structures of size less than 10 with less than 5 unsatisfied check nodes are LSSs of 6-cycles.

\subsubsection{$d_l = 5$ and $g = 8$}

\begin{theorem}
For left-regular graphs with $d_l = 5$ and $g = 8$, the multiplicity of non-isomorphic LSS$_x$ structures for different values
of $x$ are listed in Table~\ref{d5g8} for different classes of ETSs in $\mathcal{T}$.
\end{theorem}

\begin{table}[!h]
\caption{\small{LSS properties of non-isomorphic structures of $(a,b)$ ETS Classes for left-regular graphs with $d_l = 5$ and $g = 8$.  }}
\centering
\begin{tabular}{||c"c|c|c|c|c|c|c|}
\hline
\cline{1-4}
&\small{\textbf{$a\le7$}}&\small{\textbf{$a=8$}}&\small{\textbf{$a=9$}}\\\hline
\cline{1-4}
\small{\textbf{$b=0$}}&\small{-}&\small{-}&\small{-} \\\hline
\small{\textbf{$b=1$}}&\small{-}&\small{-}&\small{-}  \\\hline
\small{\textbf{$b=2$}}&\small{-}&\small{-}&\small{-}  \\\hline
\small{\textbf{$b=3$}}&\small{-}&\small{-}&\small{-}  \\\hline
\small{\textbf{$b=4$}}&\small{-}&\small{-}&\small{-}  \\\hline
\small{\textbf{$b=5$}}&\small{-}&\small{-}&\small{TS: $\begin{Bmatrix}  g \\1 \end{Bmatrix}$} \\
&\small{-}&\small{-}&\small{AS: $\begin{Bmatrix}  g \\1 \end{Bmatrix}$} \\\hline

\small{\textbf{$b=6$}}&\small{-}&\small{-}&\small{-}   \\\hline

\small{\textbf{$b=7$}}&\small{-}&\small{-}&\small{TS: $\begin{Bmatrix}  g\\1  \end{Bmatrix}$}\\
&\small{-}&\small{-}&\small{AS: $\begin{Bmatrix}  g \\1 \end{Bmatrix}$}\\\hline

\small{\textbf{$b=8$}}&\small{-}&\small{TS: $\begin{Bmatrix}  g \\1 \end{Bmatrix}$}&\small{-} \\
&\small{-}&\small{AS: $\begin{Bmatrix}  g \\1 \end{Bmatrix}$}&\small{-} \\\hline

\small{\textbf{$b=9$}}&\small{-}&\small{-}&\small{TS: $\begin{Bmatrix}  g\\3  \end{Bmatrix}$}  \\
&\small{-}&\small{-}&\small{AS: $\begin{Bmatrix}  g\\2  \end{Bmatrix}$}  \\\hline

\cline{1-4}
\end{tabular}
\label{d5g8}
\end{table}

Tables~\ref{d4g8} and \ref{d5g8}, in comparison with their counterparts for $g = 6$, show that in graphs with larger girth,
small trapping sets with small number of unsatisfied check nodes cannot exist. In particular, Table~\ref{d5g8} indicates that
in left-regular graphs with $d_l = 5$ and $g = 8$ no $(a,b)$ ETS with $a \leq 7$ and $b < 10$ can exist. The inspection of Table~\ref{d5g6}
however, reveals that many of these ETS classes can exist in left-regular graphs with $d_l = 5$ and $g = 6$.
The results of Table~\ref{d5g8} also demonstrate that, for a left-regular graph with $d_l = 5$ and $g = 8$,  all the $(a,b)$ ETSs
with $a < 10$ and $b < 10$ are LSSs of 8-cycles.

\begin{remark}
It can be proved that it is not possible to have any $(a,b)$ ETS with $a < 10$ and $b < 10$ in a left-regular graph with $d_l = 5$ and $g > 8$.
\end{remark}

\subsubsection{$d_l = 6$ and $g = 6$}

\begin{theorem}
For left-regular graphs with $d_l = 6$ and $g = 6$, the multiplicity of non-isomorphic LSS$_x$ structures for different values
of $x$ are listed in Table~\ref{d6g6} for different classes of ETSs in $\mathcal{T}$.
\end{theorem}

\begin{table}[!h]
\caption{\small{LSS properties of non-isomorphic structures of $(a,b)$ ETS Classes for left-regular graphs with $d_l = 6$ and $g = 6$.  }}
\centering
\renewcommand{\arraystretch}{1.1}
\begin{tabular}{||c"c|c|c|c|c|c|c|}
\hline
\cline{1-7}
&\small{\textbf{$a=4$}}&\small{\textbf{$a=5$}}&\small{\textbf{$a=6$}}&\small{\textbf{$a=7$}}&\small{\textbf{$a=8$}}&\small{\textbf{$a=9$}} \\\hline
\cline{1-7}
\small{\textbf{$b=0$}}&\small{-}&\small{-}&\small{-}&\small{TS: $\begin{Bmatrix}  g \\1 \end{Bmatrix}$}&\small{TS: $\begin{Bmatrix}  g \\1 \end{Bmatrix}$}&\small{TS: $\begin{Bmatrix}  g \\4  \end{Bmatrix}$}  \\
\small{\textbf{}}&\small{-}&\small{-}&\small{-}&\small{AS: $\begin{Bmatrix}  g \\1 \end{Bmatrix}$}&\small{AS: $\begin{Bmatrix}  g \\1 \end{Bmatrix}$}&\small{AS: $\begin{Bmatrix}  g \\4 \end{Bmatrix}$}  \\ \hline
\small{\textbf{$b=1$}}&\small{-}&\small{-}&\small{-}&\small{-}&\small{-}&\small{-}  \\\hline
\small{\textbf{$b=2$}}&\small{-}&\small{-}&\small{-}&\small{TS: $\begin{Bmatrix}  g \\1   \end{Bmatrix}$}&\small{TS: $\begin{Bmatrix}  g \\3 \end{Bmatrix}$}&\small{TS: $\begin{Bmatrix}  g\\25  \end{Bmatrix}$}   \\
\small{\textbf{}}&\small{-}&\small{-}&\small{-}&\small{AS: $\begin{Bmatrix}  g \\1   \end{Bmatrix}$}&\small{AS: $\begin{Bmatrix}  g \\3 \end{Bmatrix}$}&\small{AS: $\begin{Bmatrix}  g\\25  \end{Bmatrix}$}   \\ \hline

\small{\textbf{$b=3$}}&\small{-}&\small{-}&\small{-}&\small{-}&\small{-}&\small{-}  \\\hline
\small{\textbf{$b=4$}}&\small{-}&\small{-}&\small{-}&\small{TS: $\begin{Bmatrix}  g \\2  \end{Bmatrix}$}&\small{TS: $\begin{Bmatrix}  g \\15 \end{Bmatrix}$}&\small{TS: $\begin{Bmatrix}  g \\162 \end{Bmatrix}$}   \\
\small{\textbf{}}&\small{-}&\small{-}&\small{-}&\small{AS: $\begin{Bmatrix}  g \\2  \end{Bmatrix}$}&\small{AS: $\begin{Bmatrix}  g \\15 \end{Bmatrix}$}&\small{AS: $\begin{Bmatrix}  g \\162 \end{Bmatrix}$}   \\\hline

\small{\textbf{$b=5$}}&\small{-}&\small{-}&\small{-}&\small{-}&\small{-}&\small{-}  \\\hline
\small{\textbf{$b=6$}}&\small{-}&\small{-}&\small{TS: $\begin{Bmatrix}  g \\1  \end{Bmatrix}$}&\small{TS: $\begin{Bmatrix}  g \\5 \end{Bmatrix}$}&\small{TS: $\begin{Bmatrix}  g\\48  \end{Bmatrix}$}& \small{TS: $\begin{Bmatrix}  g\\726  \end{Bmatrix}$}  \\
\small{\textbf{}}&\small{-}&\small{-}&\small{AS: $\begin{Bmatrix}  g \\1  \end{Bmatrix}$}&\small{AS: $\begin{Bmatrix}  g \\5 \end{Bmatrix}$}&\small{AS: $\begin{Bmatrix}  g\\48  \end{Bmatrix}$}& \small{AS: $\begin{Bmatrix}  g\\726  \end{Bmatrix}$}  \\\hline

\small{\textbf{$b=7$}}&\small{-}&\small{-}&\small{-}&\small{-}&\small{-}&\small{-}  \\\hline
\small{\textbf{$b=8$}}&\small{-}&\small{-}&\small{TS: $\begin{Bmatrix}  g \\ 1  \end{Bmatrix}$}&\small{TS: $\begin{Bmatrix}  g \\ 10 \end{Bmatrix}$}&\small{TS: $\begin{Bmatrix}  g \\120 \end{Bmatrix}$}& \small{TS: $\begin{Bmatrix} \small  g&g+4 \\ 2273&1  \end{Bmatrix}$} \\
\small{\textbf{}}&\small{-}&\small{-}&\small{AS: $\begin{Bmatrix}  g \\ 1  \end{Bmatrix}$}&\small{AS: $\begin{Bmatrix}  g \\ 10 \end{Bmatrix}$}&\small{AS: $\begin{Bmatrix}  g \\120 \end{Bmatrix}$}& \small{AS: $\begin{Bmatrix}  g \\1157 \end{Bmatrix}$} \\ \hline

\small{\textbf{$b=9$}}&\small{-}&\small{-}&\small{-}&\small{-}&\small{-}&\small{-}  \\\hline

\small{\textbf{$b=10$}}&\small{-}&\small{TS: $\begin{Bmatrix}g \\1 \end{Bmatrix} $}&\small{TS: $\begin{Bmatrix}g \\2 \end{Bmatrix}$}&\small{TS: $\begin{Bmatrix}g \\20 \end{Bmatrix}$}&\small{TS: $\begin{Bmatrix}g \\ 260 \end{Bmatrix}$}& \small{TS: $\begin{Bmatrix} g&g+2&g+4&\text{NA}\\ 5406&2&2&1 \end{Bmatrix}$} \\

\small{\textbf{}}&\small{-}&\small{AS: $\begin{Bmatrix}g \\1 \end{Bmatrix} $}&\small{AS: $\begin{Bmatrix}g \\2 \end{Bmatrix}$}&\small{AS: $\begin{Bmatrix}g \\20 \end{Bmatrix}$}&\small{AS: $\begin{Bmatrix}g \\ 260 \end{Bmatrix}$}& \small{AS: $\begin{Bmatrix} g\\ 1620 \end{Bmatrix}$} \\

\cline{1-7}
\end{tabular}
\label{d6g6}
\end{table}

The results of Table~\ref{d6g6} indicate that for left-regular graphs with $d_l = 6$ and $g = 6$, all the possible $(a,b)$ absorbing sets with
$a < 10$ and  $b \le 10$ satisfy the LSS property with respect to 6-cycles. Moreover, except for one structure (out of 5411 structures) of the ETS class $(9,10)$,
all the structures of $(a,b)$ ETS classes with $a < 10$ and $b \leq 10$ are LSSs of short cycles. (The structure that does not satisfy the LSS property is shown in Figure~\ref{910_6}.)
Among these classes (excluding the $(9,10)$ class), only one structure out of the 2274 possible structures for the $(9,8)$ ETS class is an LSS of 10-cycles.
All the other possible structures for all the ETS classes are LSSs of 6-cycles.

\begin{figure}[!h]
\centering
\includegraphics[width=2.5in]{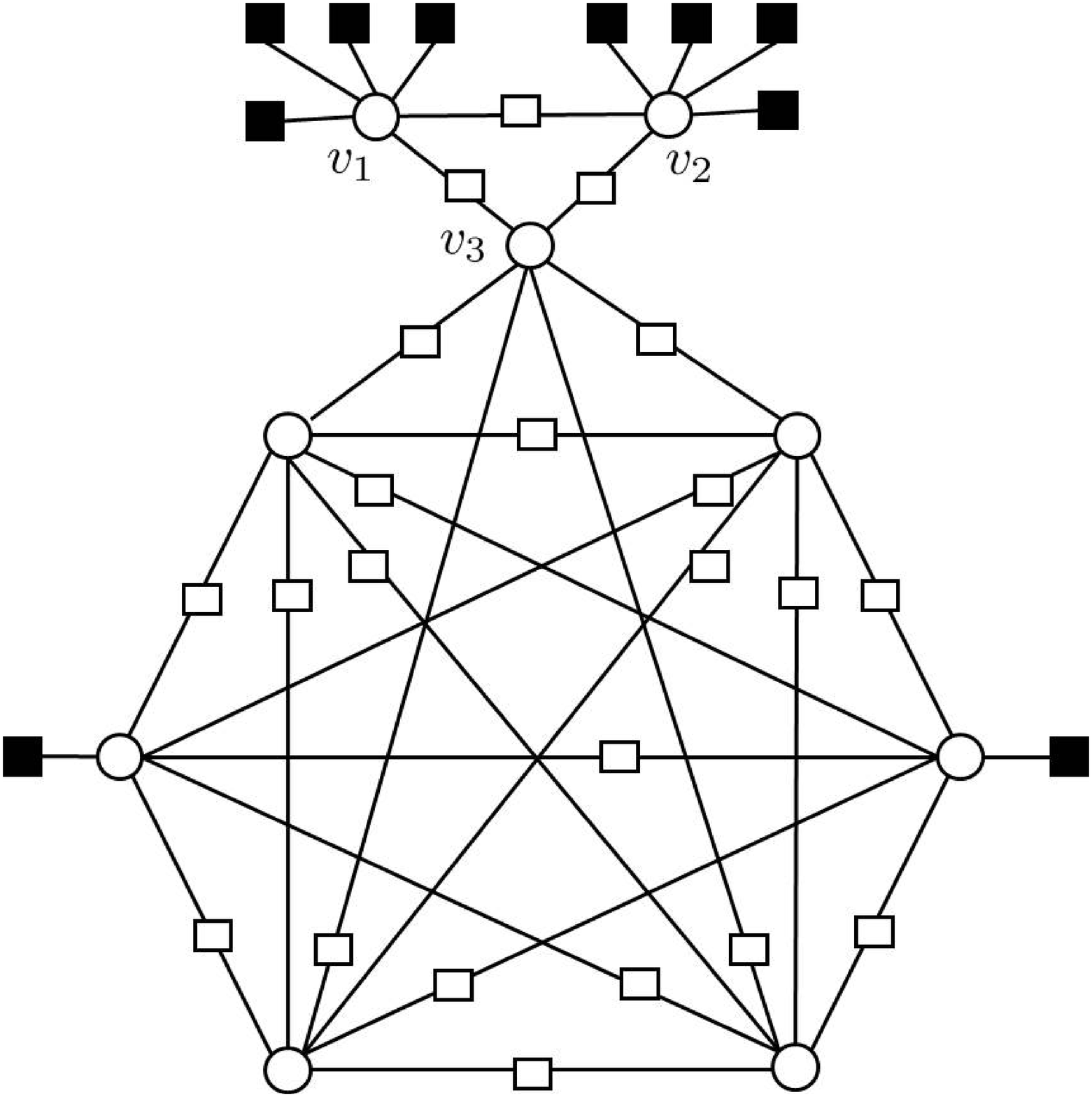}
\caption{The only possible structure for a $(9,10)$ ETS in a left-regular graph with $d_l = 6$ and $g = 6$ that does not satisfy the LSS property. }
\label{910_6}
\end{figure}

\subsubsection{$d_l = 6$ and $g > 6$}

\begin{theorem}
For left-regular graphs with $d_l = 6$ and $g > 6$, there does not exist any $(a,b)$ ETS in ${\cal T}$ with $a < 10$ and $b \le 10$.
\end{theorem}

\section{Conclusion}

In this paper, we investigated the structure of elementary trapping sets (ETS) of left-regular LDPC codes. We developed an approach
to find {\em all} the non-isomorphic structures of a given $(a,b)$ class of ETSs, where $a$ is the size and $b$ is the number of unsatisfied
check nodes of the ETS. For left-regular LDPC codes with left degrees $d_l = 3, 4, 5, 6$, and girths $g = 6, 8$, we studied
such structures and demonstrated that an overwhelming majority of them are layered supersets (LSS) of short cycles
in the Tanner graph of the code. In particular, we proved that for any category of left-regular LDPC codes with given $d_l$ and $g$,
there exist integers $\alpha$ and $\beta$ such that all the classes of $(a,b)$ ETSs with $a < \alpha$ and $b < \beta$, are
LSSs of short cycles. This implies that for any category of left-regular LDPC codes, the dominant ETSs are all
LSSs of short cycles. The LSS characterization of dominant ETSs is particularly important as it corresponds to a simple
algorithm that can find {\em all} such ETSs in a {\em guaranteed} fashion starting from the short cycles of the graph.
For any class of $(a,b)$ ETSs, the lengths of the required short cycles were provided in this paper.

One important contribution of this paper is the approach developed to exhaustively find all the
non-isomorphic structures of a given class of $(a,b)$ ETSs for arbitrary values of $a$ and $b$ and for left-regular LDPC codes
of arbitrary left degree and girth. In a more general context, the database of such structures can be
very helpful in the analysis and the design of LDPC codes with low error floor.
In particular, one can use this information to find all the ETSs of a certain class in a guaranteed fashion
regardless of whether those ETSs satisfy the LSS property or not. To the best of our knowledge,
the results presented in Tables~\ref{d3g6} -~\ref{d6g6} are the most comprehensive results available so far on the structure of
ETSs of regular LDPC codes.





\begin{thebibliography}{1}

\bibitem{ABAIT2011}
R. Asvadi, A. H. Banihashemi and M. Ahmadian-Attari, ``Lowering the error floor of LDPC codes using cyclic liftings," \emph{IEEE Trans. Inform. Theory,} vol. 57, no. 4, pp. 2213 - 2224, Apr. 2011.

\bibitem{Abu2010}
S. Abu-Surra, D. DeClercq, D. Divsalar, and W. Ryan, ``Trapping set enumerators for specific LDPC codes," Proc. \emph{Inform. Theory and Applications
Workshop,} San Diego, CA, Jan. 31- Feb. 5, 2010, pp. 1-5.


\bibitem{Abu2008}
S. Abu-Surra, W. E. Ryan, D. Divsalar, and W. Ryan, ``Ensemble Trapping Set Enumerators for
Protograph-Based Generalized LDPC Codes," Proc. \emph{Inform. Theory and Applications
Workshop,} San Diego, CA, Jan. 27�Feb. 1 2008, pp. 63�65.



\bibitem{Cavus2005}
E. Cavus and B. Daneshrad, ``A performance improvement and error floor avoidance technique for belief propagation decoding of LDPC
codes," Proc. {\em 16th IEEE International Symposium Pers. Indoor Mobile Radio Communications,} Los Angeles, CA, Sept. 2005, vol. 4, pp. 2386-2390.


%
%

\bibitem{Cole2008}
C. Cole, S. Wilson, E. Hall, and T. Giallorenzi, ``A general method for finding low error rates of LDPC codes," {\em CoRR, arxiv.org/abs/cs/0605051}.

\bibitem{Diao2012}
Qiuju Diao, Ying Yu Tai, Shu Lin, and Khaled A. S. Abdel-Ghaffar, ``Trapping set structure of finite geometry LDPC codes," Proc.
\emph{IEEE International Symposium on Information Theory
(ISIT'12)}, 2012, pp. 3088-3092.

\bibitem{Dolecek2007_1}
L. Dolecek, Z. Zhang, V. Anantharam, M.Wainwright, and B. Nikolic,
``Analysis of absorbing sets for array-based LDPC codes," Proc. {\em Interntional Conf. Communications,} Glasgow, Scotland, Jun. 24-28, 2007, pp. 6261-6268.

\bibitem{Dolecek2007}
L. Dolecek, Z. Zhang, M. Wainwright, V. Anantharam, and B. Nikolic, ``Evaluation of the
low frame error rate performance of LDPC codes using importance sampling," Proc. {\em IEEE Inform. Theory Workshop},
Lake Tahoe, CA, Sept. 2-6, 2007, pp. 202-207.

\bibitem{Dolecek2009}
L. Dolecek, P. Lee, Z. Zhang, V. Anantharam, B. Nikolic, and M. Wainwright, ``Predicting error floors of LDPC
codes: deterministic bounds and estimates,"  {\em IEEE Journal on Selected Areas in Communications},
vol. 27, no. 6, pp. 908-917, Aug. 2009.

\bibitem{Dolecek2010}
L. Dolecek, Z. Zhang, V. Anantharam, M. J. Wainwright, and B.
Nikolic, ``Analysis of absorbing sets and fully absorbing sets of
array-based LDPC codes," {\em IEEE Trans. Inform. Theory}, vol. 56, no. 1, pp. 181-201, Jan. 2010.


\bibitem{Di2002}
C. Di, D. Proietti, I. E. Telatar, T. J. Richardson, and R. L. Urbanke, ``Finite-length analysis of
low-density parity-check codes on the binary erasure channel," {\em IEEE Trans. Inform. Theory},
vol. 48, no. 6, pp. 1570-1579, June 2002.

\bibitem{Han2008}
Y. Han, and W. E. Ryan, ``LDPC decoder strategies for achieving low error floors," Proc. {\em Inform. Theory and Applications Workshop}, San Diego, CA, Jan. 2008, pp. 277-286.


\bibitem{Huang2011}
Qin Huang, Qiuju Diao, Shu Lin, and Khaled A. S. Abdel-Ghaffar, Trapping sets of structured LDPC codes," Proc.
\emph{IEEE International Symposium on Information Theory
(ISIT'11)}, 2011, pp. 1086-1090.


%





\bibitem{Ivkovic2008}
M. Ivkovic, S. K. Chilappagari, and B. Vasic, ``Eliminating trapping
sets in low-density parity-check codes by using Tanner graph covers,"
{\em IEEE Trans. Inform. Theory}, vol. 54, no. 8, pp. 3763-3768, Aug. 2008.



\bibitem{Laendner2009}
S. Laendner, T. Hehn, O. Milenkovic and J.B. Huber, ``The trapping redundancy of linear block codes,"
{\em IEEE Trans. Inform. Theory}, vol. 55, no. 1, pp. 53-63, Jan. 2009.

\bibitem{Laendner2005}
S. Laendner and O. Milenkovic, ``Algorithmic and combinatorial analysis of trapping sets in structured LDPC codes," Proc. {\em IEEE International Conference on Wireless Networks, Communications and Mobile Computing}, Hawaii, USA, June 13-16, 2005, pp. 630-635.

\bibitem{Laendner2010}
S. Laendner, T. Hehn, O. Milenkovic and J.B. Huber, ``Characterization of Small Trapping Sets in LDPC Codes from Steiner Triple Systems," Proc. {\em 6th International Symposium on Turbo Codes \& Iterative Information Processing,} Brest, France, Sept. 6 - 10, 2010, pp. 93-97.





\bibitem{KB-12}
M. Karimi and A. H. Banihashemi, ``Efficient Algorithm for Finding Dominant Trapping Sets of LDPC Codes," {\em IEEE Trans. Inform. Theory}, vol. 58, pp. 6942- 6958, Nov. 2012.



\bibitem{Kliewer2009}
C. Koller, A. Graell i Amat, J. Kliewer, and D. J. Costello, ``On trapping sets for repeat accumulate accumulate codes," Proc. {\em Inform. Theory and Applications Workshop}, San Diego, La Jolla, CA, Feb. 8-13, 2009.


\bibitem{Koller2009}
C. Koller, A. Graell i Amat, J. Kliewer, and D. J. Costello, ``Trapping
set enumerators for repeat multiple accumulate code ensembles," Proc.
\emph{IEEE International Symposium on Information Theory
(ISIT'09)}, Seoul, Korea, Jun. - Jul. 2009.



\bibitem{Koetter2004}
R. Koetter, W.-C.W. Li, P.O. Vontobel, J.L. Walker, ``Pseudo-codewords of cycle codes via zeta functions,"  Proc.
\emph{IEEE Inform. Theory Workshop,} San Antonio, TX, USA, 2004, pp. 7-12.




\bibitem{KW-12}
G. B. Kyung and C.-C. Wang, ``Finding the exhaustive list of small fully absorbing sets and designing the
corresponding low error-floor decoder," {\em IEEE Trans. Communications}, vol. 60, pp. 1487-1498, June. 2012.




\bibitem{Mao-Banihashemi-ICC-01}
Y. Mao and A.H. Banihashemi, ``A heuristic search for good low-density parity-check codes at short block lengths,"  Proc. {\em IEEE International Conference on Communications}, Helsinki, Finland, June 2001, pp. 41-44.


\bibitem{McGregor2008}
A. McGregor and O. Milenkovic, ``On the hardness of approximating stopping and trapping sets," {\em IEEE Trans. Inform. Theory}, vol. 56, no. 4, pp. 1640-1650, Apr. 2010.

\bibitem{Mackay2002}
D. J. C. MacKay and M. S. Postol, ``Weaknesses of Margulis and Ramanujan-Margulis low-density parity-check codes,"
 {\em Electronic Notes in Theoretical Computer Science},  vol. 74, 2003.


\bibitem{Milenkovict2007}
O. Milenkovic, E. Soljanin, and P. Whiting, ``Asymptotic spectra of trapping sets in regular and irregular LDPC code ensembles," {\em IEEE Trans. Inform. Theory}, vol. 53, no. 1, pp. 39-55, Jan. 2007.

%





\bibitem{Richardson2003}
T. Richardson, ``Error floors of LDPC codes," Proc. {\em  41th Annual Allerton Conference  on Communication, Control and
Computing,} Monticello, IL, Oct. 2003, pp. 1426-1435.







\bibitem{Rosnes2009}
E. Rosnes and O. Ytrehus, ``An efficient algorithm to find all small-size stopping sets of low-density parity-check matrices,"
{\em IEEE Trans. Inform. Theory}, vol. 55, no. 9, pp. 4167-4178, Sept. 2009.

\bibitem{Pusane2009}
E. Pusane, D. J. Costello, and D. G.M. Mitchell, ``Trapping set analysis
of protograph-based LDPC convolutional codes," Proc.
\emph{IEEE International Symposium on Information Theory
(ISIT'09)}, Seoul, Korea, Jun. - Jul. 2009, pp. 561�565.


\bibitem{Vasic2009}
B. Vasic, S. Chilappagari, D. Nguyen, and S. Planjery, ``Trapping set
ontology," Proc. \emph{ 47th Annual Allerton Conference  on Communication, Control and
Computing,} Monticello, IL, Sept. 2009, pp. 1-7.

%
\bibitem{Schlegel2010}
C. Schlegel and S. Zhang, ``On the Dynamics of the Error Floor Behavior in (Regular) LDPC Codes," \emph{IEEE
Trans. Inform. Theory,} vol. 56, no. 7, pp. 3248-3264, June 2010.

%




%
%
%
%
\bibitem{Vontobel2005}
P. O. Vontobel and R. Koetter, ``Graph-cover decoding and finite-length analysis of message-passing iterative
decoding of LDPC codes," {\em CoRR, arxiv.org/abs/cs/0512078}.
%


\bibitem{Wang2007}
C.-C. Wang, ``On the exhaustion and elimination of trapping sets: Algorithms \& the suppressing effect,"
Proc. \emph{ IEEE International Symposium on Inform. Theory}, Nice, France, June 24-29, 2007, pp. 2271-2275.

\bibitem{Wang2009}
C.-C. Wang, S.R. Kulkarni, H.V. Poor, ``Finding all small error-prone substructures in LDPC codes,"
{\em IEEE Trans. Inform. Theory}, vol. 55, no. 5, pp. 1976-1998, May 2009.


\bibitem{XB-07}
H. Xiao and A. H. Banihashemi, ``Estimation of bit and frame error rates of low-density parity-check codes on binary
symmetric channels," {\em IEEE Trans. Communications}, vol. 55, pp. 2234-2239, Dec. 2007.



\bibitem{XB-09}
H. Xiao and A. H. Banihashemi, ``Error rate estimation of low-density parity-check codes on binary symmetric channels using
cycle enumeration," {\em IEEE Trans. Communications}, vol. 57, pp. 1550-1555, June 2009.


\bibitem{XBK-12}
H. Xiao, A. H. Banihashemi and M. Karimi, ``Error Rate Estimation of Low-Density Parity-Check Codes Decoded by Quantized Soft-Decision Iterative Algorithms," {\em IEEE Trans. Communications}, vol. 61, pp. 474-483, Feb. 2013.


%


\bibitem{Zhang2009}
Y. Zhang and W. E. Ryan, ``Toward low LDPC-code floors: a case study,"
{\em IEEE Trans. Communications}, vol. 57, no. 6, pp. 1566-1573, June 2009.

%

\bibitem{Online}
[Online] Available: http://cs.anu.edu.au/~bdm/nauty/


\end{thebibliography}
\end{document}